\documentclass[a4paper,11pt,oneside]{article}
\usepackage[a4paper,left=2cm,right=2cm,top=2.5cm,bottom=3cm]{geometry}
\usepackage{amssymb, mathrsfs, amsthm}
\usepackage[centertags]{amsmath}
\usepackage{graphicx}
\usepackage{subcaption}
\usepackage{multicol}
\usepackage{xcolor}
\usepackage{csquotes}
\usepackage{multirow}
\usepackage{mathtools}
\usepackage{bbm}
\usepackage{bm}
\usepackage{mwe}
\usepackage[noblocks,affil-sl]{authblk}             % per avere le info degli autori
\usepackage{caption}
\usepackage{subcaption}
\usepackage{hhline}
\usepackage{epstopdf}
\usepackage{adjustbox}
\usepackage{dsfont}
\usepackage{colortbl}
\usepackage[authoryear]{natbib}
\usepackage{hyperref}
\usepackage{chngcntr}
\usepackage{geometry}
\usepackage{pdflscape}
\usepackage{cases}
\usepackage{cancel}
\usepackage{float}
\usepackage{makecell}

\counterwithin{figure}{section}
\counterwithin{table}{section}

\theoremstyle{plain}
\newtheorem{thm}{Theorem}[section]
\newtheorem{lem}[thm]{Lemma}
\newtheorem{prop}[thm]{Proposition}
\newtheorem{cor}[thm]{Corollary}
\newtheorem{defn}[thm]{Definition}
\newtheorem{ass}{Assumption}[section]
\newtheorem{remark}[thm]{Remark}

\newcommand\de{\mathrm{d}}

\numberwithin{equation}{section}

\mathtoolsset{showonlyrefs}
\makeatother
\begin{document}
    \author[2]{Katia Colaneri\thanks{Corresponding author: katia.colaneri@uniroma2.it}}
    \author[2]{Daniele Mancinelli\thanks{%Corresponding author:
    daniele.mancinelli@uniroma2.it}}
    \author[1]{Immacolata Oliva\thanks{immacolata.oliva@uniroma1.it}}
    \affil[1]{Department of Methods and Models for Economics, Territory and Finance, Sapienza University of Rome.}
    \affil[2]{Department of Economics and Finance, University of Rome Tor Vergata.}
    \title{On the optimal design of a new class of proportional portfolio insurance strategies in a jump-diffusion framework}
    \date{\today}
    \maketitle
    \begin{abstract}
    In this paper, we investigate an optimal investment problem associated with proportional portfolio insurance (PPI) strategies in the presence of jumps in the underlying dynamics. PPI strategies enable investors to mitigate downside risk while still retaining the potential for upside gains. This is achieved by maintaining an exposure to risky assets proportional to the difference between the portfolio value and the present value of the guaranteed amount. While PPI strategies are known to be free of downside risk in diffusion modeling frameworks with continuous trading, see e.g., \cite{cont2009constant}, real market applications exhibit a significant non-negligible risk, known as gap risk, which increases with the multiplier value.
    The goal of this paper is to determine the optimal PPI strategy in a setting where gap risk may occur, due to downward jumps in the asset price dynamics. We consider a loss-averse agent who aims at maximizing the expected utility of the terminal wealth exceeding a minimum guarantee. Technically, we model agent's preferences with an S-shaped utility functions to accommodate the possibility that gap risk occurs, and address the optimization problem via a generalization of the martingale approach that turns to be valid under market incompleteness in a jump-diffusion framework. 
    \end{abstract}
    \section{Introduction}\label{sect:Intro}
    Portfolio insurance (PI) strategies were devised by \cite{rubinstein1976evolution} and \cite{brennan1976pricing} after the stock market collapse in $1973$-$1974$ that caused the withdrawal of pension funds, and came back in popularity after the great financial crisis of $2008$. Nowadays, they represent a linchpin of the asset management industry for institutional such as insurance companies, mutual funds or pension funds, and retail investors. PI strategies aim to secure a predetermined minimum level of wealth over a given time horizon and simultaneously guarantee equity market participation in the event of market upturns, see, e.g., \cite{grossman1989portfolio} and \cite{basak2002comparative} for further details. For these reasons PI strategies are frequently used by pension funds’ managers to protect the contributions paid by workers during the accumulation phase of, e.g., defined contribution pension funds.
    Depending on the financial instruments that compose these portfolios and their rebalancing rules, PI strategies can be classified into \textit{option-based portfolio insurance} (OBPI) and \textit{constant proportion portfolio insurance} (CPPI) strategies. This paper focuses on a generalization of the CPPI strategy, called \textit{proportional portfolio insurance strategy}. 
    A CPPI strategy allocates investor's wealth in a diversified risky index and a reserve asset over time according to the following pre-specified rule. First, a \textit{floor} is defined as the present value of the guaranteed amount at maturity. The \textit{exposure} to the risky asset at each time is proportional to the difference between the current portfolio value and the floor as long as the current portfolio value is greater than the floor, and the proportionality factor is called \textit{multiplier} (see, e.g. \cite{black1989constant}). Since CPPI strategies are assumed to be self-financing, the remaining wealth is invested in the reserve asset. While the CPPI strategy uses a constant multiplier throughout the investment horizon, the \textit{proportional portfolio insurance} (PPI) strategy allows for a time-varying multiplier to better adapt the strategy's exposure to market fluctuations.\\
    In ideal markets without frictions, where portfolio rebalancing is  continuous, PI strategies are characterized as maximisers of expected utility of \textit{hyperbolic absolute risk aversion} (HARA) type, subject to the constraint that the strategy value at least replicates  the guaranteed amount at the end of the investment horizon. This problem has been explored by \cite{KINGSTON1989345, BLACK1992403, zieling2014performance} among others, under various modeling setting. We also mention \citet{bernard2016dynamic} who study optimal PPI strategies under standard expected utility preferences with fixed horizon and  forward utilities, in a diffusion market model.  However, with the introduction of frictions, either through jumps in the dynamics of the underlying as in \cite{cont2009constant}, or through trading restrictions like in \cite{balder2009effectiveness}, PI strategies fail to meet one of their objectives, potentially leading to the so-called \textit{gap risk}. In such instances, the value of the strategy drops below the floor, and therefore, according to the PI allocation rule, from that time until maturity the wealth is fully directed into the reserve asset. This results in a loss, namely the \textit{gap}, given by the difference between the final portfolio wealth and the guaranteed amount.
    One of the drawbacks of this mechanism is that the portfolio becomes fully monetized and thus incapable of capitalising on equity market participation in the event of subsequent market rises. Taking gap risk into account is of fundamental importance, for instance to avoid losses for financial institutions that issue these kinds of strategies, or to protect the capital in pension fund management. For instance, in case of defined contribution (DC) pension funds, the gap risk could drop pension funds' wealth below the desired level, coinciding with the present value of pension obligations, making them underfunded, see, e.g. \cite{temocin2018constant} for further details.  A discussion on the performance of standard OBPI and CPPI strategies within the framework of DC pension funds management, in presence of downside risk, is given in \citet{xu2020portfolio}. Hence, several gap risk hedging methodologies have been proposed in the literature, which spam from the use of a new class of exotic options, called \textit{gap options}, proposed by \cite{tankov2010pricing}, to the introduction of a \textit{conditional multiplier}  (see \cite{hamidi2009caviar, HAMIDI20141, ameur2014portfolio, ameur2018risk, dichtl2017bootstrap}), that incorporates expectations of future market drops.\\ 
    Due to their popularity, PI strategies are also analysed using elements of behavioural finance. In particular, \cite{DICHTL20111683} carried out an extensive empirical analysis and showed that PI strategies can also be explained using the \textit{cumulative prospect theory} (CPT) developed by \cite{tversky1992advances}. This study highlights PI investors are loss-averse and endowed with an utility function, that is concave for gains and convex for losses. Put in other words, portfolio insurers evaluate the investment outcome based on deviation from a reference point, for instance the guaranteed amount, and assess potential gains and losses asymmetrically. More precisely, they evaluate the marginal utility of potential losses higher than that of potential gains.
    
    In the CPT framework, \cite{ELZ2020} study the optimal design of a variant of the CPPI strategy, called \textit{generalized behavioural portfolio insurance} (GBPI) strategy which is well defined even in case of gap risk. 
    After such an event, the GBPI maintains positive exposure to the risky assets aiming to mitigate losses and reach the funded area. The authors show optimality of the proposed GBPI strategy in a CPT framework, assuming discrete-time rebalancing, diffusive dynamics for the risky asset, and constant risk-free rate.\\
    Our paper extends the idea of \cite{ELZ2020} to the PPI strategy. We propose a portfolio allocation rule that maintains equity exposure even when the cushion becomes negative, hence under gap risk. To achieve this, we consider a PI insurer who aims to maximize the expected utility of the cushion at maturity. We model the insurer's preferences using an S-shaped utility function, which reflects varying levels of risk aversion for gains and losses. We consider a setting that includes frictions in financial insurance market, via downward jumps in the dynamics of the underlying risky asset. We solve the optimization problem using a procedure based on the following two steps. First, we show the equivalence between the optimization problem with an S-shape utility function with that of its concave envelope. Then we apply a modification of the  \emph{martingale method}, based on duality theory, where the original dynamic problem is transformed into an equivalent static problem. 
    Here the main difficulty arises due to market incompleteness, which does not allow for a univocal characterization of the state price density. Hence, we rely on a technique developed by \cite{michelbrink2012martingale} which combines the martingale method and \emph{worst case} probability, and allows to determine the PPI strategy, as well as the state price density, by solving a system of non-linear equations. We discuss a case where the optimized multiplier is unique and given in quasi-closed form. We also perform an accurate numerical analysis, using several specifications of jump size distributions, providing interesting results. In particular, we show that the optimized multiplier of the proposed PPI strategy prevents gap risk during the entire investment horizon.\\
    The remainder of the paper is organized as follows. In Section \ref{sect:PPI}, we briefly recap the state of the art on optimized PPI strategies. In Section \ref{sect:Financial_mkt_model}, we describe the market model with jumps. The optimization problem is introduced in Section \ref{sec:optimization_pb}, also containing a discussion about the concavification method. Section \ref{sec:martingale_approach} determines the optimal PPI strategy and discusses a case where a (semi-)closed form of the optimal multiplier can be established. We perform a numerical analysis in Section \ref{sect:numeric}, and Section \ref{sect:conclusion} poses the conclusions. Technical details are provided in the appendices.
    
    \section{Proportional Portfolio Insurance strategies and gap risk}\label{sect:PPI} 
    PPI strategies are designed to capitalize on the returns of a risky asset $S$, such as a market index, through a dynamic trading strategy while ensuring a guaranteed minimum amount $G$ at a maturity $T$. To accomplish this objective, the fund manager divides her position between a risky asset $S=\left(S_t\right)_{t\in[0,T]}$ and a reserve asset, typically a zero-coupon bond or a money market account with price process $P=\left(P_t\right)_{t\in[0,T]}$, %and the same maturity $T$, 
    according to the following procedure.  Initially, the manager defines a floor process $F=(F_t)_{t\in[0,T]}$ and the cushion process $C=(C_t)_{t\in[0,T]}$. The floor represents the protected capital amount at each point in time $t\in [0,T]$, and is given by the expected present value of the guaranteed amount at maturity. The cushion is the difference between the  portfolio value $V$ and the floor $F$, i.e. at every time $t \in [0,T]$, it satisfies $C_{t}=V_{t}-F_{t}$.
    The exposure to the risky asset is consistently linked to the cushion. Indeed, at any given time $t \in [0,T]$, if $V_{t}>F_{t}$, the manager invests $m_{t}\cdot C_{t}$ in the risky asset, with the multiplier $m=(m_{t})_{t\in[0,T]}$ determining the proportionality factor. However, if $V_{t}\le F_{t}$ at some $t< T$, then a complete shift of the portfolio value into the reserve asset until maturity is prompted. In summary, the exposure to the risky asset in a PPI strategy is expressed as $m_{t}\cdot\left(C_{t}\right)^{+}$ for all $t \in [0,T]$. In a standard setup, PPI insurers aim to determine the multiplier to maximize the expected hyperbolic absolute risk aversion (HARA)\footnote{A HARA utility function is given by $U^{HARA}(x)=\frac{(x-G)^{1-\gamma}}{1-\gamma}$,   
    for $\gamma>0$, $\gamma\neq 1$, and $x>G\ge 0$.}
    utility of terminal portfolio value under the constraint that the terminal portfolio value exceeds or equals the guaranteed amount $G$ (see e.g. \cite{KINGSTON1989345}). 
    Equivalently, the problem can be rephrased in terms of the cushion: the optimal multiplier maximizes the expected constant relative risk aversion (CRRA) utility function of the form $U^{CRRA}(x)=\frac{x^{1-\gamma}}{1-\gamma}$, of the terminal cushion, under the constraint that the terminal cushion is non-negative. This optimization problem has been solved by \cite{zieling2014performance} taking the assumption that the market is frictionless and the dynamics of the underlying risky asset and the reserve asset follow geometric Brownian motions. Under this setup, however, the problem simplifies because the constraint for the terminal cushion is not binding, and the optimization can be solved via standard arguments. The optimal multiplier turns out to be non-constant and presents two components: the myopic demand expressed in terms of the Merton solution and a strongly model-dependent intertemporal hedging demand. Under the special choice of constant drift and volatility of the underlying risky asset and constant risk-free interest rate, optimal multiplier becomes constant and given by the instantaneous Sharpe ratio, properly scaled with the risk aversion parameter. In this case, the PPI strategy collapses into the Constant Proportion Portfolio Insurance (CPPI) strategy. A model of this type, however, fails to capture frictions in market dynamics, such as liquidity shortages or unexpected shocks that induce sudden price changes. Accounting for these features would necessitate considering market dynamics that eventually exhibit unpredictable downward jumps.  \cite{cont2009constant} show that, when jumps are included, the value of a PPI strategy is affected by gap risk with positive probability, under certain parameter configurations. Technically, it means that there exists a critical time $\tau:=\inf\left\lbrace t\ge 0:V_{t}\leq F_{t}\right\rbrace \wedge T$, so that under a standard PPI mechanism, from $\tau$ on, the residual wealth is entirely invested into the reserve asset, since $(C_t)^+=0$ for all $t \in [\tau \wedge T, T]$. To illustrate the impact of gap risk on portfolio insurers, we have conducted a historical simulation of the CPPI strategy applied to the Standard and Poor 500 index, spanning from 2006 to 2013, utilizing a constant multiplier of 10. The results, provided in Figure \ref{fig:Historical_simulation}, reveal that the CPPI strategy fails to adequately mitigate the risk posed by the sudden market collapse in 2008, leading to a drop in its value, which falls below the predetermined floor. Consequently, the remaining wealth is entirely shifted to the risk-free asset, making the strategy incapable of securing the guaranteed amount at maturity. Most importantly, although the Standard and Poor 500 index nicely recovers from early 2009 onwards, the presence of gap risk prevents the CPPI strategy from engaging in any equity market participation post-2008.
    \begin{figure}[th!]
    \centering
    \includegraphics[width=0.75\textwidth]{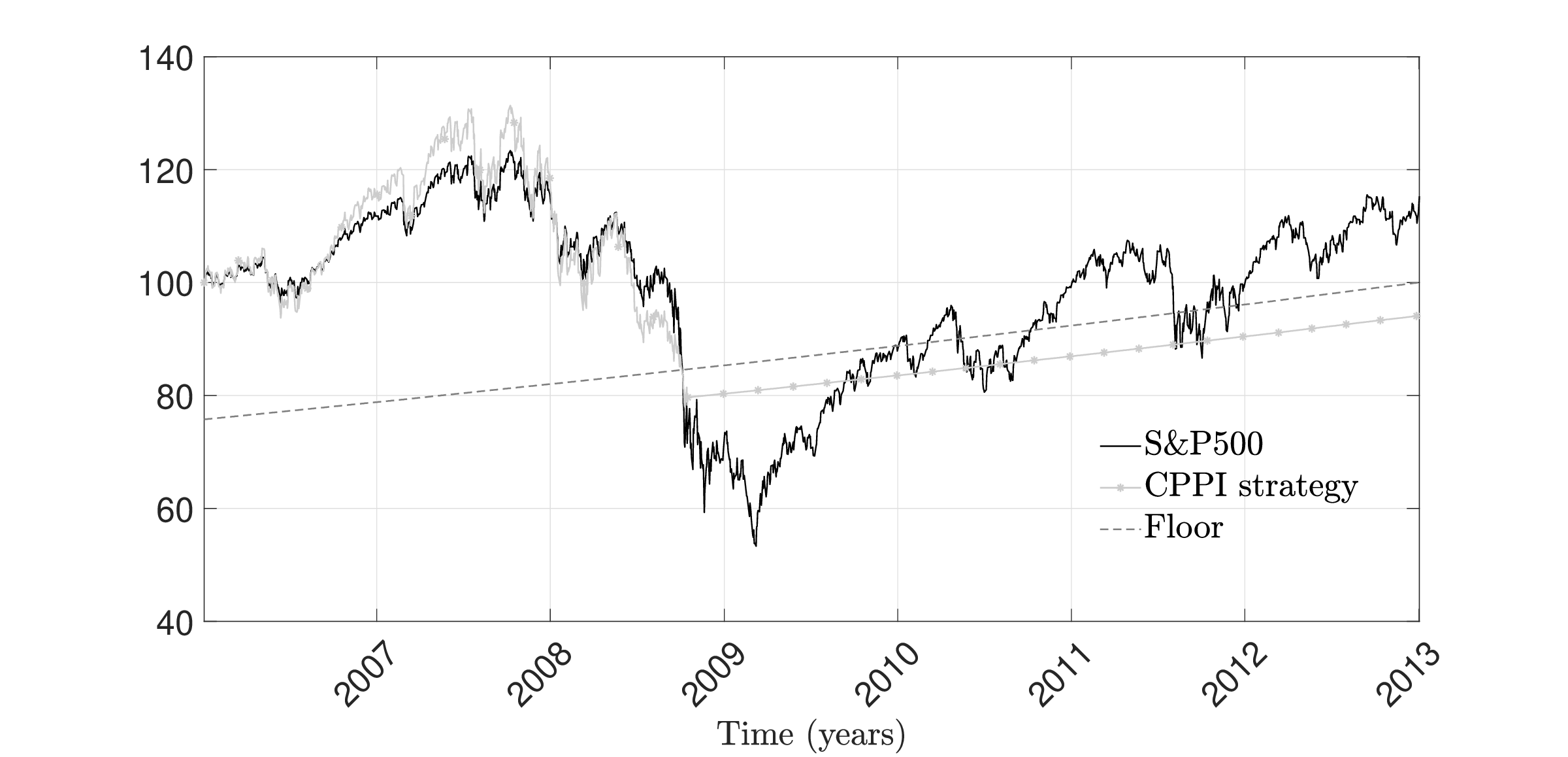} 
    \caption{Historical simulations for the CPPI strategy on the S\&P 500 index. The solid black line represents the trajectory of the S\&P500 index, the solid grey line is the value of the CPPI strategy with constant multiplier equal to 10, and the dashed grey line is the value of the floor with constant interest rate equal to $3.5$.\%} 
    \label{fig:Historical_simulation}
    \end{figure}
     In this paper, we propose a modification of the PPI algorithm that guarantees equity market participation over the entire investment time horizon, even when gap risk occurs, eventually leading to a smaller loss at maturity.
     
    \section{The financial market model}\label{sect:Financial_mkt_model}
    We introduce the mathematical framework. Let $(\Omega, \mathcal F, \mathbb{P})$ be a probability space and let $T$ be a finite time horizon. Consider a one-dimensional  Brownian motion $(W_t)_{t \ge 0}$ and a jump process described by a homogeneous Poisson random measure $N\left(\de t,\de y\right)$, independent of each other. We let $\nu\left(\de y\right)$ be the $\sigma$-finite compensator of $N(\de t, \de y)$, with support in a measurable set $E \subseteq \mathbb{R}$. Let $\mathbb F=\{\mathcal{F}_{t},\;t \in [0,T]\}$ be a complete and right continuous filtration given by 
    \begin{equation*}
    \mathcal{F}_{t}:=\sigma\left\lbrace\left(W_s, N\left((0,s],A\right)\right)| \ A\in\mathcal{B}(E), \ s\in[0,t]\right\rbrace\vee\mathcal{N},
    \end{equation*}
    with $\mathcal{F}_T=\mathcal F$ and $\mathcal{F}_{0}$ being the trivial $\sigma$-algebra, where $\mathcal{N}$ is the collection of all $\mathbb{P}$-null sets and $\mathcal{B}(E)$ is the Borel-$\sigma$-algebra on $E$. Next, we consider a financial market which consists of a money market account with value process $P=(P_t)_{t\in[0,T]}$ and a stock with price process $S=(S_t)_{t\in[0,T]}$. We assume that a fund manager can buy and sell continuously without restrictions or transaction costs. We assume that the interest rate is constant and equal to $r>0$, so that the value of the money market account is $P_t=e^{rt} $. The stock dynamics is assumed to follow a geometric jump-diffusion process, that is 
    \begin{equation}\label{eq:stock_price}
    \frac{\de S_t}{S_{t-}}=\mu\de t+\sigma\de W_t+\int_E\gamma(t,y)N(\de t,\de y),\quad S_{0}= s_0>0,
    \end{equation}
    where  $\mu\in\mathbb{R}$ is the drift of the risky asset, $\sigma\in\mathbb{R}_{+}$ is the volatility of the risky asset, and $\gamma(t,y)$ is a continuous function accounting for the size of jumps in the stock price. To ensure non-negativity of the risky asset price we assume that $\gamma(t,y)>-1$ for all $(t,y)\in [0,T]\times E$.  
    Note that that equation \eqref{eq:stock_price} admits a unique strong solution.\\
    We consider a self-financing trading strategy $\bm{\pi}=\left(1-\pi_t, \pi_t\right)_{t\in[0,T]}$ where $\pi_t$ is the fraction of portfolio value invested in the risky asset at time $t$ (consequently, $1-\pi_t$ is the fraction of wealth invested in the bond at time $t$). Hence, the dynamics of the wealth process ${V}^{\bm{\pi}}=({V}^{\bm{\pi}}_{t})_{t\in[0,T]}$ is given by 
    \begin{equation*}
    \frac{\de{V}^{\bm{\pi}}_{t}}{{V}^{\bm{\pi}}_{t-}}=\left[r+\pi_t\left(\mu-r\right)\right]\de t+\pi_t\sigma\de W_t+\pi_t\int_E\gamma(t,y)N(\de t,\de y), 
    \end{equation*}
    with ${V}^{\bm{\pi}}_{0}={v}_0$ being the initial endowment. We further assume standard integrability conditions on $\bm{\pi}$, that is
    \begin{equation}\label{eq:integrability_cond}
    \mathbb{E}\left[\int_0^T\pi_s^2\left(1+\int_E\gamma^2(s,y)\nu(\de y)\right) \de s\right]<\infty.
    \end{equation}
    Note that condition \eqref{eq:integrability_cond} simplifies in $\mathbb{E}\left[\int_0^T \pi_s^2 \de s\right]<\infty$ if $\mathrm{supp} (\nu)=E$ is a compact subset of $\mathbb{R}$. Note also that $\pi_t$ can assume values in $\mathbb{R}$ with the convention that a negative $\pi_t$ corresponds to short-selling of the risky asset and $\pi_t>1$ correspond to borrowing from the bank account. 
    We let $G$ be the guaranteed amount at maturity. In typical agreements $G=\xi\,v_0$, where $\xi\in (0,1]$ is some pre-determined \textit{protection level}; hence $G$ is set at the initial time. We define the floor $F=(F_t)_{t\in[0,T]}$ as the present value of the guaranteed amount, i.e. $F_t=Ge^{-r(T-t)}$ and, for every strategy $\bm{\pi}$, we assume that the exposure to the risky asset is proportional to the difference between the portfolio value $V^{\bm{\pi}}$ and the floor $F$. We denote by $m=(m_t)_{t\in[0,T]}$ the proportionality factor, and then in mathematical terms, we have $\pi_t V_t^{\bm{\pi}}=m_t C^{m}_t$ for every $t\in [0,T]$, where ${C}^{m}_{t}:={V}^{\bm{\pi}}_{t}-F_{t}$ indicates the cushion, regardless of gap risk. Note that there is a one to one correspondence between the strategy $\bm{\pi}$ and the multiplier $m$, therefore we can equivalently use the superscript $V^{m}$ when expressing the portfolio value in terms of the multiplier, that is
    \begin{align}\label{eq:PPI_portfolio_value}
    \de{V}^{m}_{t}=r{V}^{m}_{t-}\de t+m_t\left({V}^{m}_{t-}-F_{t}\right)\left[\left(\mu-r\right)\de t+\sigma\de W_t+\int_E\gamma(t,y)N(\de t,\de y)\right],\quad{V}^{m}_{0}={v}_0,
    \end{align}
    and the dynamics of the cushion process are
    \begin{align}\label{eq:cushion_process}
    \frac{\de{C}^{m}_{t}}{{C}^{m}_{t-}}=\left[r+m_{t}\left(\mu-r\right)\right]\de t+m_{t}\sigma\de W_t+m_{t}\int_E\gamma(t,y)N(\de t,\de y),\quad {C}^{m}_{0}={c}_0.
    \end{align}
    We note that the integrability condition \eqref{eq:integrability_cond} can be rewritten in terms of the multiplier as   
    \begin{equation}\label{eq:integrability_cond_mult}
    \mathbb{E}\left[\int_0^Tm^2_s\left(1+\int_E\gamma^2(s,y)\nu(\de y) \right) \de s\right]<\infty.
    \end{equation}
    Hence we can give the following definition of admissible multipliers.
    \begin{defn}
    An admissible multiplier is an $\mathbb{F}$-predictable process $m=(m_{t})_{t\in[0,T]}$ such that
    \begin{itemize}
    \item[(i)] The integrability condition \eqref{eq:integrability_cond_mult} is satisfied;
    \item[(ii)] ${C}^{m}_{t}\ge -F_{t}$ $\mathbb{P}$-a.s. for all $t\in[0,T]$, and for every initial cushion ${c}_0={v}_0- F_0>0$.\footnote{Equivalently, an admissible strategy $\bm{\pi}$ satisfies $V^{\bm{\pi}}_{t}\geq 0$, for all $t\in[0,T]$, and for every initial wealth ${v}_{0}>0$.}
    \end{itemize}
    \end{defn}
    We denote by $\mathcal{M}$ the set of all admissible multipliers $m$.\footnote{This set is non-empty, as an investor can always invest all his/her funds into the money market account, that is $m_{t}=0$ for all $t\in[0,T]$.} By stretching the definition slightly, we refer to the strategy described above as PPI. Moreover,  to help readability we drop the superscript if there is no ambiguity in the notation, or if we do not need to stress the dependence of the processes on the strategy.
    \section{The optimization problem}\label{sec:optimization_pb}
    The presence of jumps in the dynamics of the underlying risky asset leads to gap risk, even under broader definitions of PPI strategies. The cushion process may get negative values which does not permit modelling the preferences of the insurer via a CRRA utility function. Therefore, in our work, we adopt the prospect utility theory to model the preference of the PI insurer. Motivated by the work of \cite{DICHTL20111683}, we assume that insurer's preferences are described by an S-shaped utility function as follows (see also \cite{bian2011smooth} and \cite{DONG2020341}):  
    \begin{equation}\label{eq:S_shaped_utility_fun} 
    U(x;G)=
    \begin{cases}
    \begin{aligned}
    &\dfrac{\left(x-G\right)^{1-\delta_1}}{1-\delta_1}, \quad x\geq G,\\
    &-\tilde{\lambda}\dfrac{\left(G-x\right)^{1-\delta_2}}{1-\delta_2},\quad 0\leq x<G,
    \end{aligned}
    \end{cases}
    \end{equation}
    where $\delta_i\in\left(0,1\right)$, for $i=1,2$ and $\tilde{\lambda}>\frac{1-\delta_2}{1-\delta_1}$, and the guaranteed amount $G$ is the reference point. Hence, the fund's manager aims to determine the optimal multiplier of our new version of the PPI strategy in order to maximize the expected S-shaped utility of the terminal wealth ${V}^{m}_{T}$. 
    Due to the special choice of the utility function, we do not restrict ourselves to considering strategies leading to positive cushion. Even if the gap occurs during the investment horizon, the fund manager maintains the exposure to the risky asset proportional to the negative cushion. 
    Hence, we recast the optimization problem in terms of the cushion as follows
    \begin{equation}\label{eq:opt_problem_2}
    \mbox{Maximize }\mathbb{E}^{t,{c}}\left[\tilde{U}\left({C}^{m}_{T};G\right)\right]\mbox{ over all }m\in\mathcal{M},
    \end{equation}
    where $\mathbb{E}^{t,{c}}$ denotes the conditional expectation given ${C}_{t}={c}$, and $\tilde{U}(\cdot\mbox{ };G)$ is an S-shaped utility on $[-G,\infty)$ with reference point equal to zero, defined as
    \begin{equation*} 
    \tilde{U}(x;G)=
    \begin{cases}
    \begin{aligned}
    &\frac{x^{1-\delta_1}}{1-\delta_1}, \quad x\geq 0,\\
    &-\tilde{\lambda}\frac{\left(-x\right)^{1-\delta_2}}{1-\delta_2},\quad -G\leq x<0.
    \end{aligned}
    \end{cases}
    \end{equation*}
    From now on, we omit the dependence of the multiplier $m$ in the notation of the cushion process to help for readability. We define the value function corresponding to the problem \eqref{eq:opt_problem_2} as follows
    \begin{equation}\label{eq:value_2}
    v(t,{c}):=\sup_{m\in\mathcal{M}\left({c}\right)}\mathbb{E}^{t,{c}}\left[\tilde{U}\left({C}_{T};G\right)\right].    
    \end{equation}
    \subsection{The concavification}\label{sect:Opt_problem}
    Since the utility function in equation \eqref{eq:value_2} is  not concave,  we consider its concave envelope and show that the optimum of problem \eqref{eq:value_2} is achieved when $\tilde U(x)$ and the concave envelope coincide (see, e.g., \cite{carpenter2000does} for additional references on this procedure).  We recall that for a function $f$ with domain $D$, the concave envelope $f^{con}$ is given by
    \begin{equation*}\label{eq:def_concavification}
    f^{con}:=\inf\left\lbrace g:D\rightarrow\mathbb{R}\mbox{ s.t. }g \mbox{ is concave, }g(x)\geq f(x),\mbox{ }\forall x\in D\right\rbrace.   
    \end{equation*}
    Figure \ref{fig:Payoff_Diagram} provides a representation of the function $\tilde{U}$ (solid line) and its concavification $\tilde{U}^{con}(\cdot\mbox{ };G)$ (dotted line).
    \begin{figure}[H]
    \centering
    \includegraphics[width=0.7\textwidth]{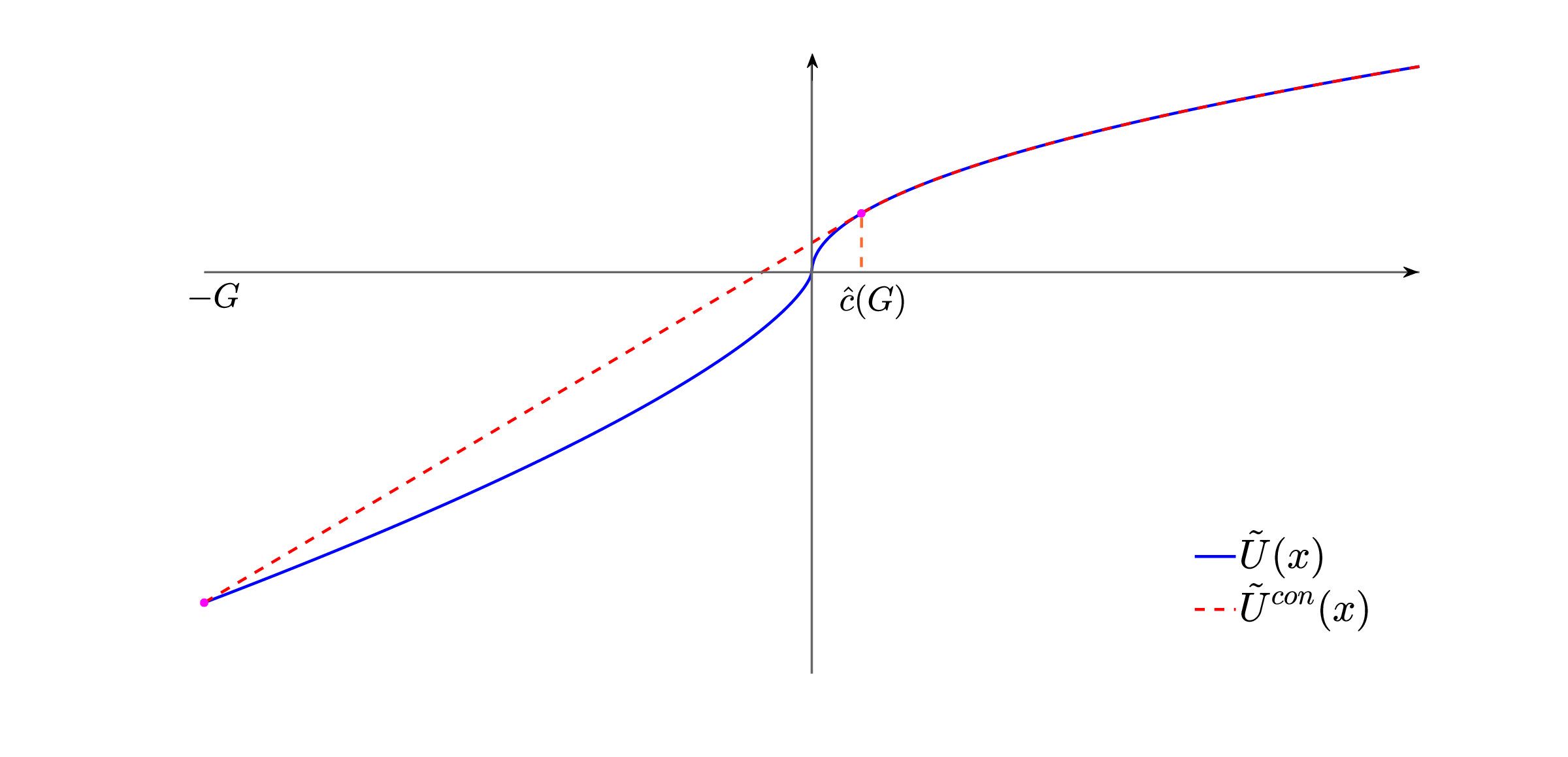} 
    \caption{The solid line plots the utility function $\tilde U(x)$, the dotted line is the concave envelope of $\tilde{U}(x)$, $\tilde{U}^{con}(x)$. In this plot we have used the parameters: $\delta_1=0.3$, $\delta_2=0.5,$ $\tilde{\lambda}=2.25$ and $G=100.$} 
    \label{fig:Payoff_Diagram}
    \end{figure}
    In particular, $\tilde{U}^{con}$ is obtained by replacing part of the original function with a straight line from $-G$ to a point, $\hat{c}(G)>0$, at which the slope of such a straight line equals the slope of $\tilde{U}$. The following lemma shows that for every $G>0$, the point $\hat{c}(G)>0$ exists and is unique.
    \begin{lem}\label{prop:existence_uniqueness_point_c_hat}
    For every $G>0$, there exists a unique point $\hat{c}(G)>0$ which is the solution of the following equation
    \begin{equation}\label{eq:tangent_point}
    \frac{\tilde{U}\left(\hat c(G);G\right)-\tilde{U}(-G;G)}{\hat c(G)+G}=\tilde{U}^\prime\left(\hat c(G);G\right).
    \end{equation}
    \end{lem}
    \begin{proof}
    See Appendix \ref{app:A_1}.
    \end{proof}
    The explicit expression for the concavified function $\tilde{U}^{con}(\cdot;G)$ is given by
    \begin{equation*}
    \tilde{U}^{con}(x;G)
    =
    \begin{cases}
    \begin{aligned}
    &\dfrac{x^{1-\delta_1}}{1-\delta_1},\quad x\geq\hat{c}(G),\\
    &kx+\dfrac{G}{\hat{c}(G)+G}\dfrac{\hat{c}(G)^{1-\delta_2}}{1-\delta_2}-\tilde\lambda\dfrac{G^{1-\delta_1}}{\hat{c}(G)+G}\dfrac{\hat{c}(G)}{1-\delta_1},\quad -G\leq x<\hat{c}(G),
    \end{aligned}
    \end{cases}
    \end{equation*}
    where $k=\dfrac{1}{\hat{c}(G)+G}\left(\dfrac{\hat{c}(G)^{1-\delta_2}}{1-\delta_2}+\dfrac{\tilde\lambda G^{1-\delta_1}}{1-\delta_1}\right)$   is the slope of the tangent line and $\hat{c}(G)>0$ is the unique solution of equation \eqref{eq:tangent_point}. Furthermore, $\tilde{U}^{con}(x;G)\geq\tilde{U}(x;G)$ for all $(x,G)\in[-G,+\infty)\times\left(0,+\infty\right)$ and $\tilde{U}^{con}\left(x;G\right)=\tilde{U}\left(x;G\right)$ for $x=-G$ and for all $x>\hat{c}(G)$. Next, we show that the solution of the optimization is achieved at a point where the utility function $\tilde{U}$ and its concave envelop $\tilde{U}^{con}$ coincide.   
    \begin{lem}\label{lem:2}
    The optimal terminal cushion is achieved at ${C}^{\star}_T\geq\hat{c}(G)$.
    \end{lem}
    \begin{proof}
    See Appendix \ref{app:A_2}.
    \end{proof}
    The range of possible values for the optimal cushion depends on both the evaluation of gains (i.e. on the parameter $\delta_1$) and the penalization of losses (i.e. the parameters $\delta_2$ and $\tilde \lambda$). In particular, for fixed $\delta_1$, the optimal cushion assumes a range of larger values if penalization of losses is smaller.\\     
    Thanks to Lemma \ref{lem:2}, the optimization problem \eqref{eq:opt_problem_2} translates into the following equivalent one
    \begin{equation}\label{eq:opt_problem_3}
    \mbox{Maximize }\mathbb{E}^{t,{c}}\left[\tilde{U}^{con}\left({C}_{T};G\right)\right]\mbox{ over all }m\in\mathcal{M}.
    \end{equation}
    \section{A martingale approach to optimal investment under market incompleteness}\label{sec:martingale_approach}
    Incompleteness of the market impedes the resolution of optimization problem \eqref{eq:opt_problem_3} via the standard martingale method. Therefore, we use a generalization of this approach developed by \cite{michelbrink2012martingale}, that combines the martingale method with the worst case probability, suitably tailored to our setup.
    \begin{remark}
    Alternatively one could apply the dynamic programming approach, and characterize the value of the optimization as the unique viscosity solution of the Hamilton-Jacobi-Bellman (HJB) equation, or the unique classical solutions under additional conditions on the parameters of the model  (see, e.g. \cite{pham1998optimal}). Under general conditions on the dynamics of the state processes, the optimal PPI strategy can only be expressed in feedback form and the value function must be computed numerically. The approach followed in this paper permits to obtain a different characterization of the optimal PPI strategy via the solution of a non-linear system. This may be as complicated as solving the HJB equation, but it has the advantage to draw a connection between the optimal multiplier and the market price of risk.         
    \end{remark}
    Under the assumption of no arbitrage, there are infinitely many equivalent martingale measures (EMMs) that can be characterized by the family of Radon-Nikodym densities $Z^{\bm{\theta}}$ so that
    \begin{equation}\label{eq:radon_nik_der}
    \frac{\de\mathbb{Q}_{\bm{\theta}}}{\de\mathbb{P}}\bigg|_{\mathcal{F}_{T}}=Z^{\bm{\theta}}_{t},
    \end{equation}
    depending on the vector process $\bm{\theta}=\left(\theta^D,\theta^J\right)$, where $\theta^D$ and $\theta^J$ are the market prices of risk for the diffusion and for jumps of the stock price, respectively. For the predictable process $\left(\theta^D_t\right)_{t\in[0,T]}$, and the predictable random field $\left(\theta^J_t(y)\right)_{t\in[0,T]}$, we assume that the process $Z^{\bm{\theta}}$ is a martingale\footnote{Sufficient conditions on $\bm{\theta}$ will be outlined later.} and it is the solution of the SDE
    \begin{equation}
    \frac{\de Z^{\bm{\theta}}_{t}}{Z^{\bm{\theta}}_{t-}}=\theta^{D}_t\de W_t-\int_E\left(1-\theta^J_t(y)\right)\left(N(\de t,\de y)-\nu(\de y)\de t\right),
    \end{equation}
    with $Z^{\bm{\theta}}_{0}=1$. If $\theta^J_t(y)>0$ for all $y\in E$ and $t\in[0,T]$, it holds that
    \begin{multline}\label{eq:radon_nik_der_1}
    Z^{\bm{\theta}}_{t}=\exp\Bigg\{\int_0^t\theta^D_s\de W_{s}-\frac{1}{2}\int_0^t\left(\theta^D_s\right)^2\de s\\+\int_0^t\int_E\log\theta^J_s(y)N(\de s,\de y)+\int_0^t\int_E\left(1-\theta^J_s(y)\right)\nu(\de y)\de s\Bigg\},
    \end{multline}
    see, Theorem $4.61$ of \cite{jacod2013limit} for further details. By Girsanov's Theorem, if $Z^{\bm{\theta}}_{t}$ is a $\mathbb{P}$-martingale with expectation equal to one, we get that  
    \begin{equation}\label{eq:Brown_Motion_Q1}
     W^{\mathbb{Q}_{\bm{\theta}}}_t=W_t-\int_0^t\theta^D_s\de s,\quad t\in[0,T]
    \end{equation}
    is a Brownian motions under $\mathbb{Q}_{\bm{\theta}}$ and the compensator of the random measure $N$ under $\mathbb{Q}_{\bm{\theta}}$ is $\theta^J(t,y)\nu(\de y)\de t$ for every $t\in[0,T].$ 
    To summarize, we consider the following definition. 
    \begin{defn}\label{def:set_maket_price_risk}
    The vector process $\bm{\theta}=(\theta^D, \theta^J)$ is an admissible market price of risk if it is predictable and the following conditions hold
    \begin{itemize}
    \item[(i)] $\theta^J_t(y)>0$ for all $y\in E$ and $t\in[0,T]$,
    \item[(ii)] $Z^{\bm{\theta}}_{t}$ defined by \eqref{eq:radon_nik_der_1} is a martingale with expected value equal to one,
    \item[(iii)] for all $t\in[0,T]$ it holds that
    \begin{align} \label{eq:No_Arb_Cond1} 
    &\mu-r+\sigma\theta^D_t+\int_E\gamma(t,y)\theta^J_t(y)\nu(\de y)=0.
    \end{align}
    \end{itemize}
    We let $\Theta$ be the set of all admissible market prices of risk. 
    \end{defn}
    To guarantee that $Z^{\bm{\theta}}$ is a martingale with $\mathbb{E}\left[Z^{\bm{\theta}}_{t}\right]=1$, one could assume a generalization of the Novikov's condition for jump-diffusion processes, that is $\mathbb{E}\left[e^{\frac{1}{2}\int_0^T\left(\theta^D_u\right)^2\de u+\int_0^T\left(1-\theta^J_u(y)\right)^2\de u}\right]<\infty$, see, e.g. Theorem $9$ of \cite{protter2008no}. Condition $(iii)$ of Definition \ref{def:set_maket_price_risk} ensures that the discounted stock price $S_te^{-rt}$ is a local martingale under the corresponding new probability measure $\mathbb{Q}_{\bm{\theta}}$ defined by \eqref{eq:radon_nik_der}. 
    We further impose the following.
    \begin{ass}
    The set $\Theta$ is non-empty.    
    \end{ass}
    For every $\bm{\theta}\in\Theta$, using the definition of the $\mathbb{Q}_{\bm{\theta}}$-Brownian motions, the $\mathbb{Q}_{\bm{\theta}}$-compensator of the Poisson random measure and conditions  \eqref{eq:No_Arb_Cond1}, it can be easily seen that the discounted cushion process $({C}_{t}e^{-rt})_{t \in [0,T]}$ is a local $\mathbb{Q}_{\bm{\theta}}$-martingale. To obtain properties of the discounted cushion under the original probability $\mathbb{P}$, we introduce the state price density associated with $\bm{\theta}\in\Theta$. 
    \begin{defn}
    For $\bm{\theta}=\left(\theta^D,\theta^J\right)$, the process defined by $H^{\bm{\theta}}_{t}:=Z^{\bm{\theta}}_{t}e^{-rt}$, for every $t\in[0,T]$, is called the state price density associated with $\bm{\theta}$, where $Z^{\bm{\theta}}$ is defined by equation \eqref{eq:radon_nik_der_1}. The dynamics of $H^{\bm{\theta}}$ is given by
    \begin{align*}
    \de H^{\bm{\theta}}_{t}&=H^{\bm{\theta}}_{t-}\left[-r\de t+\theta^D_t\de W_t-\int_E\left(1-\theta^J_t(y)\right)\left(N(\de t,\de y)-\nu(\de y)\de t\right)\right].
    \end{align*}
    \end{defn}
    For $\bm{\theta}\in\Theta$, the process $M$ defined by $ M_{t}:={C}_{t}H^{\bm{\theta}}_{t},$ for all $t\in[0,T]$,
    is a local $\mathbb{P}$-martingale. Moreover, for all $m\in\mathcal{M}({c}_0)$, the following budget constraint holds: 
    $\mathbb{E}^{\mathbb{Q}}\left[{C}_{T}e^{-rT}\right]\leq{c}_0 $ or equivalently $ \mathbb{E}\left[{C}_{T}H^{\bm{\theta}}_{T}\right]\leq {c}_0$,
    with ${c}_{0}:={v}_{0}-F_0$, see, e.g., Proposition $3.3$ of \cite{michelbrink2012martingale}. Hence, the static optimization problem equivalent to \eqref{eq:opt_problem_3} is
    \begin{equation}\label{eq:opt_prob_3}
    \max_{{C}_T\in\mathcal{F}_{T}}\mathbb{E}\left[\tilde U^{con}\left({C}_T\right)\right], \mbox{ with budget constraint } \mathbb{E}\left[{C}_TH^{\bm{\theta}}_{T}\right]\leq{c}_{0}.
    \end{equation}
    Thanks to Lemma \ref{lem:2}, for $\bm{\theta}=\left(\theta^D,\theta^J\right)$, we define for all $y >0$, $\mathcal{X}_{\bm{\theta}}(y):=y^{-\frac{1}{\delta_1}}\mathbb{E}\left[\left(H^{\bm{\theta}}_{T}\right)^{1-\frac{1}{\delta_1}}\right]
    $, and consider the set $\tilde{\Theta}\subseteq\Theta$ of all change of measures for which $\mathcal{X}_{\bm{\theta}}(y)$ is finite.
    For a fixed ${c}_{0}>0$ and $\bm{\theta}\in\tilde{\Theta}$, we define the non-negative random variable $Y_{\bm{\theta}}$ as
    \begin{equation}\label{eq:set_solutions}
    Y_{\bm{\theta}}:=\left(\mathcal{X}_{\bm{\theta}}^{-1}(c)H^{\bm{\theta}}_{T}\right)^{-\frac{1}{\delta_1}}.
    \end{equation}
    The properties $Y_\theta$ are provided in the lemma below (see also Lemma 1 in \cite{michelbrink2012martingale}).
    \begin{lem}\label{lem:Lemma_1}
    For any $\bm{\theta}\in\tilde{\Theta}$, $Y_{\bm{\theta}}$ defined in equation \eqref{eq:set_solutions} satisfies
    \begin{itemize}
    \item[(i)] $\mathbb{E}\left[H^{\bm{\theta}}_{T}Y_{\bm{\theta}}\right]={c}_{0}$,
    \item[(ii)] $\mathbb{E}\left[\min\left(\dfrac{Y_{\bm{\theta}}^{1-\delta_1}}{1-\delta_1},0\right)\right]>-\infty$,
    \item[(iii)] $\mathbb{E}\left[\dfrac{{C}_{T}^{1-\delta_1}}{1-\delta_1}\right]\leq\mathbb{E}\left[\dfrac{Y^{1-\delta_1}_{\bm{\theta}}}{1-\delta_1}\right]$ for all $m\in\mathcal{M}\left({c}_0\right)$.
    \end{itemize}
    \end{lem}
    An immediate consequence of point $(iii)$ in Lemma \ref{lem:Lemma_1} is  that
    \begin{equation}\label{eq:inf_problem}
    \sup_{\tilde{m}\in\mathcal{M}\left({c}_0\right)}\mathbb{E}\left[\dfrac{{C}_{T}^{1-\delta_1}}{1-\delta_1}\right]\leq\inf_{\tilde{\bm{\theta}}\in\tilde{\Theta}}\mathbb{E}\left[\dfrac{Y^{1-\delta_1}_{{\tilde{\bm{\theta}}}}}{1-\delta_1}\right],
    \end{equation}
    or in other terms, the expected utility of the auxiliary process outperforms or is at least equal to the expected utility of any admissible multiplier.
    \begin{defn}
    A martingale measure $\mathbb{Q}_{\theta}$ obtained by equation \eqref{eq:radon_nik_der} in terms of a $\hat{\bm{\theta}}\in\tilde{\Theta}$ is called optimal for the infimum problem defined in equation \eqref{eq:inf_problem} if  
    \begin{equation}\label{eq:opt_martingale_measure}
    \mathbb{E}\left[\frac{Y^{1-\delta_1}_{\hat{\bm{\theta}}}}{1-\delta_1}\right]=\inf_{\bm{\theta}\in\tilde{\Theta}}\mathbb{E}\left[\frac{Y^{1-\delta_1}_{\bm{\theta}}}{1-\delta_1}\right],
    \end{equation}
    where $Y_{\bm{\theta}}$ is defined by equation \eqref{eq:set_solutions}.
    \end{defn}
    In the following, we link the optimal martingale measure $\mathbb{Q}_{\bm{\theta}}$ defined by means of equation \eqref{eq:opt_martingale_measure} with the optimal multiplier which solve \eqref{eq:opt_problem_2}. First, we consider for any $\bm{\theta}\in\tilde{\Theta}$ the martingale $M^{\bm{\theta}}$ defined as follows
    \begin{equation}\label{eq:process_m}
    M^{\bm{\theta}}_{t}:=\mathbb{E}\left[H^{\bm{\theta}}_{T}Y_{\bm{\theta}}|\mathcal{F}_{t}\right].
    \end{equation}
    such that $M^{\bm{\theta}}_{0}={c}_{0}$ $\mathbb{P}-$a.s. for every $\bm{\theta}\in\tilde{\Theta}$. Then, in terms of $M^{\bm{\theta}}$, 
    we have
    \begin{equation}\label{eq:cushion_theta}
    {C}^{\bm{\theta}}_t:=\frac{M^{\bm{\theta}}_{t}}{H^{\bm{\theta}}_{t}},\quad t\in[0,T],\quad\bm{\theta}\in\tilde{\Theta}.
    \end{equation}
    
    To draw a connection with a more standard case, in a complete market $\bm{\theta}$ is unique and hence $Y_{\bm{\theta}}$ is the optimal cushion at maturity and $C^{\bm{\theta}}$ is the cushion process. If the market is incomplete clearly this does not hold for generic $\bm{\theta}$, so the task is to identify a specific process, say $\hat{\bm{\theta}}$, whose associated random variable $Y_{\hat{\bm{\theta}}}$ and process $C^{\hat{\bm{\theta}}}$ determine the optimal cushion at maturity and cushion process. 
    
    To derive the optimal multiplier, and optimal market risk premia, the martingale representation coefficients of $M^{\bm{\theta}}$ in equation \eqref{eq:process_m} need to be computed. To do this, we first have to decompose the state price density $H^{\bm{\theta}}$ as in the following Proposition.
    
    \begin{prop}\label{propos:state_price_decomp_3}
    For all $\bm{\theta}\in\tilde{\Theta}$, it holds that $\left(H^{\bm{\theta}}_{t}\right)^{-\frac{1-\delta_1}{\delta_1}}=\dfrac{D\left(t,H^{\bm{\theta}}_{t}\right)}{\tilde{D}\left(t\right)}$, for all $t\in[0,T]$,
    where the process $\tilde{D}\left(t\right)$ is given by
    \begin{align}\label{eq:tilde_G_process_1}
    \nonumber\tilde{D}(t)=&\exp\bigg\{\dfrac{1-\delta_1}{\delta_1}r(T-t)+\int_t^T\int_E\left[-\dfrac{1-\delta_1}{\delta_1}\left(1-\theta^J_s(y)\right)+\left(\theta^J_s(y)\right)^{-\frac{1-\delta_1}{\delta_1}}-1\right]\nu(\de y)\de s\\
    &+\frac{1}{2}\dfrac{1-\delta_1}{\delta_1^2}\int_t^T\left(\theta^D_s\right)^2\de s\bigg\},\quad t\in[0,T],
    \end{align}
    and the process $D\left(t,H^{\bm{\theta}}_{t}\right)$ is a martingale given by 
    \begin{align}\label{eq:martingale_process_3}
    \nonumber D\left(t,H^{\bm{\theta}}_{t}\right)=&\exp\bigg\{-\frac{1}{2}\left(\dfrac{1-\delta_1}{\delta_1}\right)^2\int_0^t\left(\theta^{D}_s\right)^2\de s-   \int_0^t\int_E\left[\left(\theta^J_s(y)\right)^{-\frac{1-\delta_1}{\delta_1}}-1\right]\nu(\de y)\de s\\
    \nonumber&-\dfrac{1-\delta_1}{\delta_1}\int_0^t\theta^{D}_s\de W^2_{s}-\frac{1-\delta_1}{\delta_1}\int_0^t\int_E\ln\left(\theta^J_s(y)\right)N(\de s,\de y)\Bigg\},\quad t\in[0,T].
    \end{align}
    \end{prop}
    \begin{proof}
    See Appendix \ref{app:B_6}.
    \end{proof}
This result allows to characterize $\mathcal{X}_{\bm{\theta}}(x)$, hence $Y_{\theta}$, and consequently the martingale $M^{\bm{\theta}}$, as shown in the proposition below.

\begin{prop}\label{propos:martingale_rapp_theo_1}
    Let $a^D(t)$, and $a^J(t,y)$ be the martingale representation coefficients of $M^{\bm{\theta}}$ depending on $\bm{\theta}\in\tilde{\Theta}$. Then, the dynamics of $M^{\bm{\theta}}$ is given by
    \begin{equation}\label{eq:M_mrt_1}
    \de M^{\bm{\theta}}_{t}=a^D(t)\de W_t+\int_Ea^J(t,y)\left(N\left(\de t,\de y\right)-\nu\left(\de y\right)\de t\right),
    \end{equation}
    where
    \begin{align*}
    a^D(t)&=-M^{\bm{\theta}}_{t-}\frac{1-\delta_1}{\delta_1}\theta^D_s,\\
    a^J(t,y)&=M^{\bm{\theta}}_{t-}\left[\left(\theta^J_s(y)\right)^{-\frac{1-\delta_1}{\delta_1}}-1\right].
    \end{align*}
    \end{prop}
    \begin{proof}
    See Appendix \ref{app:B_2}.
    \end{proof}
    
    To close the loop we characterize both the optimal multiplier and state price density, determining in turn the worst-case martingale measure.  
    For the optimal market price of risk $\hat{\bm{\theta}}\in\tilde{\Theta}$, the corresponding process ${C}^{\hat{\bm{\theta}}}$ defined by equation \eqref{eq:cushion_theta} gives the optimal cushion.

   \begin{thm}\label{teo:thm_1}
    Suppose that there exist a $\hat{\bm{\theta}}\in\tilde{\Theta}$ and a multiplier $\hat{m}^{\hat{\bm{\theta}}}\in\mathcal{M}$ that satisfy
    \begin{equation}\label{eq:system_sol_1}
    \begin{cases}
    \begin{aligned}
    &m_t\sigma=-\dfrac{\theta^D_t}{\delta_1},\\
    &m_t\gamma(t,y)=\left(\theta^J_t(y)\right)^{-\frac{1}{\delta_1}}-1,
    \end{aligned}
    \end{cases}
    \;
    \end{equation}
    for all $y\in E$. Assume further that the SDE depicted in equation \eqref{eq:cushion_process} has a solution for $m=\hat m^{\hat{\bm{\theta}}}$. Then $\hat m^{\hat{\bm{\theta}}}$ is a solution to the optimization problem \eqref{eq:opt_problem_2}. The corresponding cushion process is given by ${C}_{t}={C}^{\hat{\bm{\theta}}}_{t}$ $\mathbb{P}-$a.s., for all $t\in[0,T]$, where ${C}^{\hat{\bm{\theta}}}$ is defined in equation \eqref{eq:cushion_theta}.
    \end{thm}
    \begin{proof}
    See Appendix \ref{app:B_3}.    
    \end{proof}    

We remark that to obtain equivalence between the solution of the dual problem and that of the primal problem we rely on Theorem 2 of \cite{michelbrink2012martingale}. 

\medskip
    
Thanks to condition \eqref{eq:No_Arb_Cond1}, we can characterize the optimal multiplier in terms of the market parameters only. That is, if the equation  
    \begin{equation} \label{eq:absence_arb}
    \mu-r-\sigma^2 \delta_1 m_t+\int_E\gamma(t,y) \left[m_{t} \gamma(t,y)+1\right]^{-\delta_1} \nu(\de y)=0,
    \end{equation}
    has a unique solution $\hat{m}^{\hat{\bm{\theta}}}_t$ for every $t\in[0,T]$, then $\hat{m}^{\bm{\hat{\theta}}}_t$ is the optimal multiplier. Note that \eqref{eq:absence_arb} is well defined only if $m_{t-}\gamma(t,y)>-1$. This condition implies in particular that the cushion is strictly positive, which agrees with the result in Lemma \ref{lem:2}.\\
    Next, we discuss specific conditions under which equation \eqref{eq:absence_arb} has a unique solution and hence the optimal multiplier is well-defined. For the sake of simplicity, we further assume that the jump size is time-independent, 
    that is $\gamma(t,y)=\gamma(y)$ for every $t\in[0,T]$. In this case, the Girsanov kernel $\bm{\theta}$ with $\theta^D_t:=\theta^D$ and $\theta^J_t(y):=\theta^J(y)$ satisfies a simplified version of the condition \eqref{eq:absence_arb} which is given by
    \begin{equation}\label{eq:Girsanov_kernel_Levy_process}
    \mu-r+\sigma\theta^D+\int_E%{\mathbb{R}\setminus\left\lbrace0\right\rbrace}
    \gamma(y)\theta^J(y)\nu\left(\de y\right)=0.
    \end{equation}
    The optimal multiplier is no longer time-varying and satisfies the following conditions
    \begin{equation}\label{eq:system}
    \begin{cases}
    \begin{aligned}
    &m\sigma=-\dfrac{\theta^D}{\delta_1},\\   
    &m\gamma(y)=\left(\theta^J(y)\right)^{-\frac{1}{\delta_1}}-1.
    \end{aligned}
    \end{cases}
    \end{equation}
    %\end{cor}
    Taking condition \eqref{eq:Girsanov_kernel_Levy_process} on $\hat{\bm{\theta}}$ into account, we get that the optimal multiplier $\hat m^{\hat{\bm{\theta}}}$ can be characterized as the solution of the equation   
    \begin{equation}\label{eq:optimal_condition_mult}
    \mu-r-\delta_1\sigma^2m+\int_E%{\mathbb{R}\setminus\left\lbrace0\right\rbrace}
    \frac{\gamma(y)}{\left(1+\gamma(y)m\right)^{\delta_1}}\nu\left(\de y\right)=0 .  
    \end{equation}
    We discuss the sufficient conditions for the existence and uniqueness of the solution of equation \eqref{eq:optimal_condition_mult}. Let $\gamma(E)\subseteq[-1,+\infty)$ be the subset in which all jumps $\gamma\left(\cdot\right)$ lie. Furthermore, denote by $\phi_1=\inf_{y\in\mathbb{R}\setminus\left\lbrace 0\right\rbrace}\gamma(y)$ and $\phi_2=\sup_{y\in\mathbb{R}\setminus\left\lbrace 0\right\rbrace}\gamma(y)$, the infimum and the supremum of the set $\Phi$, respectively. 
    \begin{prop}\label{prop:existence_uniqueness}
    Assume that the integrals
    \begin{equation}
    \int_E%{\mathbb{R}\setminus\left\lbrace 0\right\rbrace}
    \gamma(y)\left(1-\frac{\gamma(y)}{\phi_1}\right)^{-\delta_1}\nu(\de y) \mbox{ and }\int_E%{\mathbb{R}\setminus\left\lbrace 0\right\rbrace}
    \gamma(y)\left(1-\frac{\gamma(y)}{\phi_2}\right)^{-\delta_1}\nu(\de y) 
    \end{equation}
    are well-defined or, at most, are $\pm\infty$. If the model parameters are such that
    \begin{align}
    \label{eq:cond_1}\mu-r+\dfrac{\delta_1\sigma^2}{\phi_2}+\int_E%{\mathbb{R}\setminus\left\lbrace0\right\rbrace}
    \gamma(y)\left(1-\frac{\gamma(y)}{\phi_2}\right)^{-\delta_1}\nu\left(\de y\right)&\geq 0,\\
    \label{eq:cond_2}\mu-r+\dfrac{\delta_1\sigma^2}{\phi_1}+\int_E%{\mathbb{R}\setminus\left\lbrace0\right\rbrace}
    \gamma(y)\left(1-\frac{\gamma(y)}{\phi_1}\right)^{-\delta_1}\nu\left(\de y\right)&\leq 0,
    \end{align}
    then there exists a unique $\hat m^{\hat{\bm{\theta}}}$ that solves equation \eqref{eq:optimal_condition_mult}. Moreover: 
    \begin{itemize}
    \item[(i)] If $-1\le\phi_1,\,\phi_2<0$, then $\hat m^{\hat{\bm{\theta}}}\in(-\infty, -1/\phi_1)$;  
    \item[(ii)] If $0<\phi_1,\,\phi_2<+\infty$, then $\hat m^{\hat{\bm{\theta}}}\in(-1/\phi_2,+\infty)$;  
    \item[(iii)] If $-1\leq\phi_1<0<\phi_2<+\infty$, then $\hat m^{\hat{\bm{\theta}}}\in\left(-1/\phi_2,-1/\phi_1\right)$.
    \end{itemize}
    \end{prop}
    \begin{proof}
    See Appendix \ref{app:B_7}.
    \end{proof}
    \begin{remark}\label{rem:opt_multiplier_features}
    In the cases (i), (ii), and (iii) of Proposition \ref{prop:existence_uniqueness} the optimal multiplier is such that the corresponding cushion is non-negative in the whole time interval $[0,T]$, as the condition $1+\hat{m}^{\hat \theta}\gamma(y)>0$ is satisfied. Hence, gap risk does not occur, even if $S$ has downward jumps. 
    Case (i) corresponds to a model where $S$ has only negative jumps, and these jumps may be unbounded. Here the optimal multiplier may be negative, which implies negative equity exposure, equivalently short-selling the asset. From a financial point of view, if the stock price dynamics presents a downward trend, generated for instance by many negative jumps and a small or negative excess return, then short-selling the asset may turn losses into profits. In case $(ii)$ the stock price has only positive jumps, and the range of values of the multiplier that does not lead to negative cushion is bounded below by  $-1/\phi_2$, hence $\hat m_{\hat \theta}$ may potentially be negative. In particular if $\phi_2=+\infty$, there is always positive equity exposure. Note that, here that we have a different PPI mechanism that does not impose the exposure to be equal to zero when the value of the portfolio equals the floor,  hence cash-lock does not occur. 
    In case (iii) the stock price $S$ has both positive and negative jumps, but the negative jumps are bounded. Form a modeling perspective this corresponds to assume that plausibly, the stock cannot suddenly loose more than a certain finite percentage of its value. The optimal multiplier takes values over a bounded interval. If $\phi_2=+\infty$, $\hat m_{\hat \theta}$ is positive, which guarantees positive equity exposure.    
    \end{remark}
    \section{Numerical analysis}\label{sect:numeric}
    In this section, we propose a few numerical experiments where the results of Proposition \ref{prop:existence_uniqueness} are applied to compute the optimal multiplier $\hat m^{\hat{\bm{\theta}}}$. 
    We consider different specifications of the jump size distributions and we also run a sensitivity analysis on model parameters. 
    We compare three cases: $(1)$ constant jump size; $(2)$ double exponential distribution of the jump sizes, namely, the Kou's model; $(3)$ Gaussian distribution of the jump sizes, i.e.,  the Merton's model.

    \subsection{Optimal multiplier under a jump-diffusion model with constant jump size}
    In our first numerical example we assume that the dynamics of the risky asset price process follows a jump-diffusion and presents only downward jumps. In particular, we assume that  \begin{equation}\label{eq:const_negative_jump_size}
    \frac{\de S_{t}}{S_{t-}}=\mu\de t+\sigma\de W_t+\tilde{\gamma}\de N(t),
    \end{equation}
    where $N$ is an homogeneous Poisson process with intensity rate $\lambda>0$, and $\tilde{\gamma}\in[-1,0)$ is the constant negative jump size. This example serves as a simplification of the general case, resulting in a readily computable multiplier. However, it has the following interpretation: $\tilde \gamma$ can be viewed as the largest plausible loss, which could be estimated for instance based on some risk measure, and hence, the hedging problem of the guaranteed amount is studied under a prudential perspective. The optimal multiplier $\hat m^{\hat{\bm{\theta}}}$ is determined in the corollary below.
    \begin{cor}\label{cor:constant_jump_size}
    In a market model where the risky asset price process is described by equation \eqref{eq:const_negative_jump_size}, the optimal multiplier $\hat m^{\hat{\bm{\theta}}}\in\left(-\infty,- 1/\tilde{\gamma}\right)$ is the unique solution of 
    \begin{equation}\label{eq:optimal_mult_constant_jump_size}
    \mu-r-\delta_1\sigma^2m+\lambda\tilde{\gamma}\left(1+\tilde{\gamma}m\right)^{-\delta_1}=0.
    \end{equation}
    \end{cor}
    \begin{proof}
    See Appendix \ref{app:C_1}.
    \end{proof}
    We fix the model parameters according to Table \ref{tab:cons_jump_size}. Then, we vary each of them (one by one) to run a sensitivity analysis on the optimal multiplier.
    \begin{table}[ht]
    \centering
    \begin{tabular}{cccccc} 
    \hline
    $\mu-r$  & $\sigma$ & $\lambda$ & $\tilde{\gamma}$ & $\delta_1$ & $\hat{m}^{\hat{\bm{\theta}}}$\\ 
    \hline
    $0.20$   & $0.30$     & $11$     & $-0.03$ & $0.60$ & $-2.18$\\
    \hline
    \end{tabular}
    \caption{Parameters for the jump-diffusion model with negative constant jump size}\label{tab:cons_jump_size}
    \end{table}
    The numerical analysis is relatively easy to obtain for the constant jump size, due to the quasi-explicit form of the multiplier. Nevertheless, they provide a useful comparison case study. Figure \ref{fig:constant} shows that the multiplier is decreasing with respect to the jump intensity, and it is increasing in the excess return of the underlying risky asset, the volatility of the underlying risky asset, the risk aversion level and the jump size. What is interesting and deviates from the classical results on PPI in the diffusion setting is that the multiplier may be negative. Since the optimal multiplier does not allow for a negative cushion (see Remark \ref{rem:opt_multiplier_features}), then having $\hat m^{\hat{\bm{\theta}}}<0$ implies that the fund manager should short-sell the risky asset. This behaviour occurs for low values of the volatility of the risk asset and the investor's risk aversion. When such two parameters increase, the multiplier increases, reducing the short-selling position and the overall riskiness of the strategy. More interestingly, short-selling also occurs when the intensity of negative jumps increases, the jump size (in absolute value) becomes prohibitive, or the excess return gets close to zero (or negative). These two latter effects push down the value of the risky asset, and hence, short-selling the asset may contrast with such a downward trend. 
    \begin{figure}[ht!]
    \centering
    \includegraphics[width=0.32\linewidth]{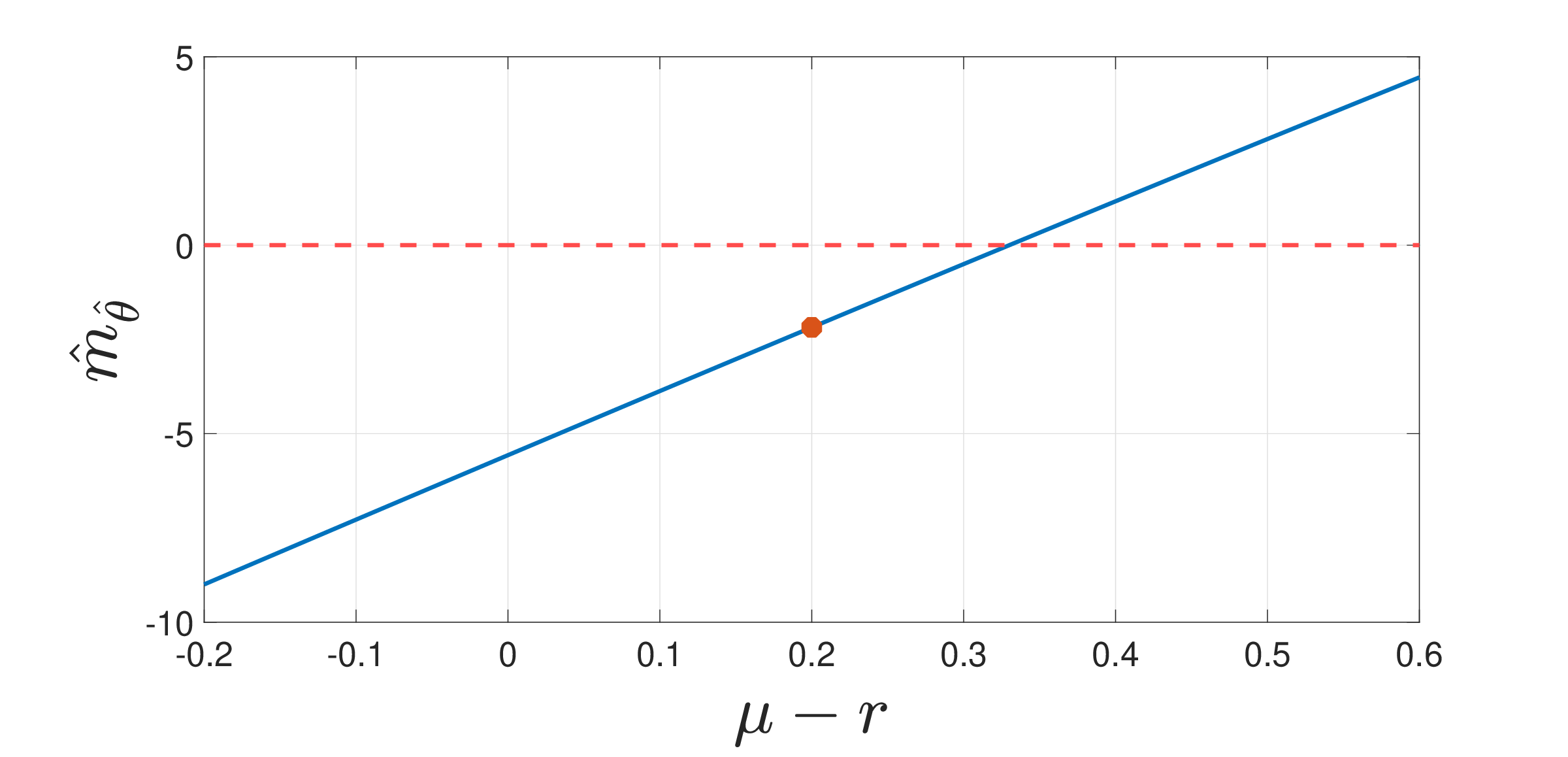}
    \hspace{0.001\textwidth}
    \includegraphics[width=0.32\linewidth]{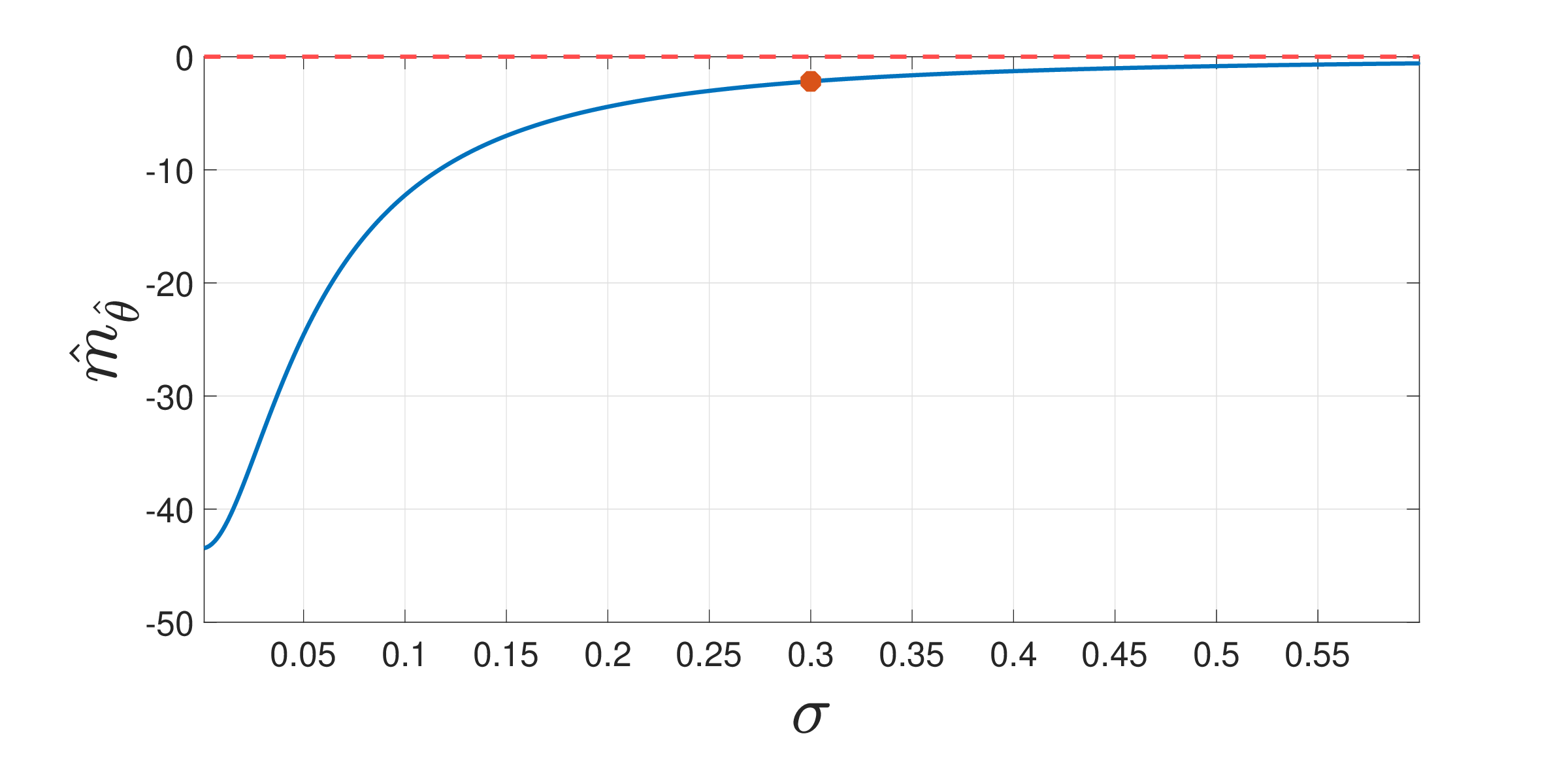}
    \hspace{0.001\textwidth}
    \includegraphics[width=0.32\linewidth]{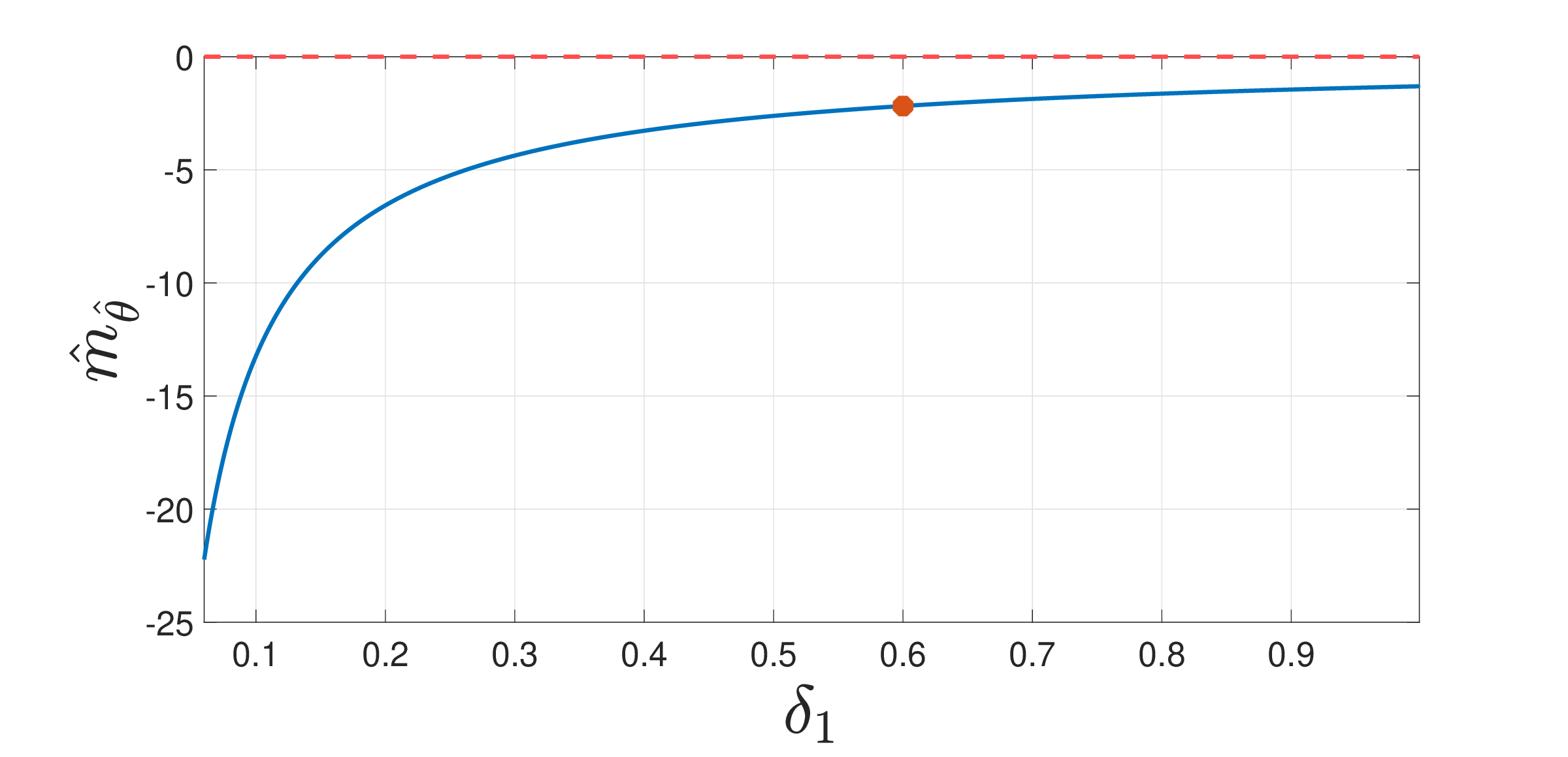}
    \vspace{0.6cm} % Spazio verticale tra le righe di figure
    \centering
    \includegraphics[width=0.32\linewidth]{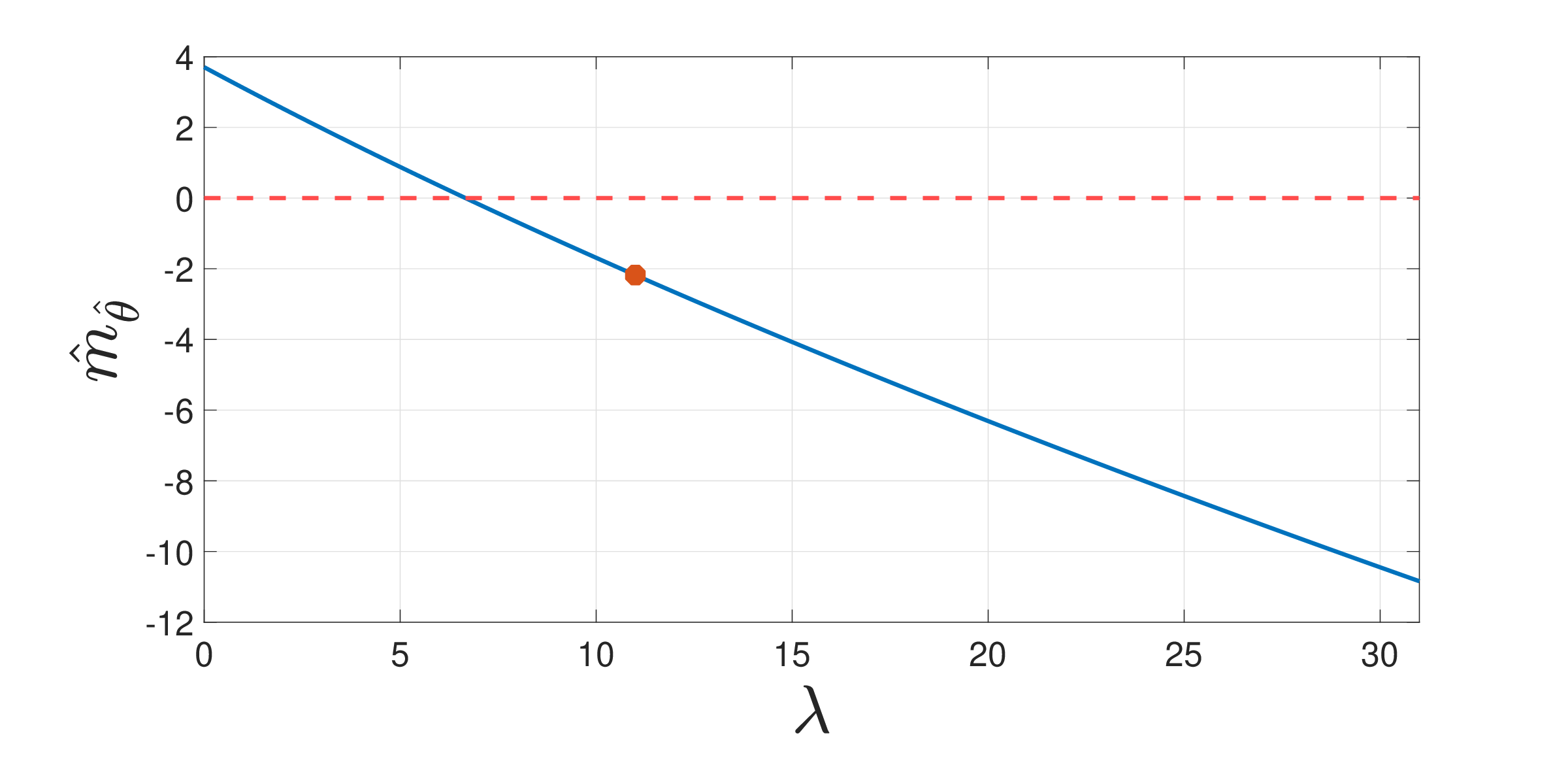}
    \hspace{0.01\textwidth}
    \includegraphics[width=0.32\linewidth]{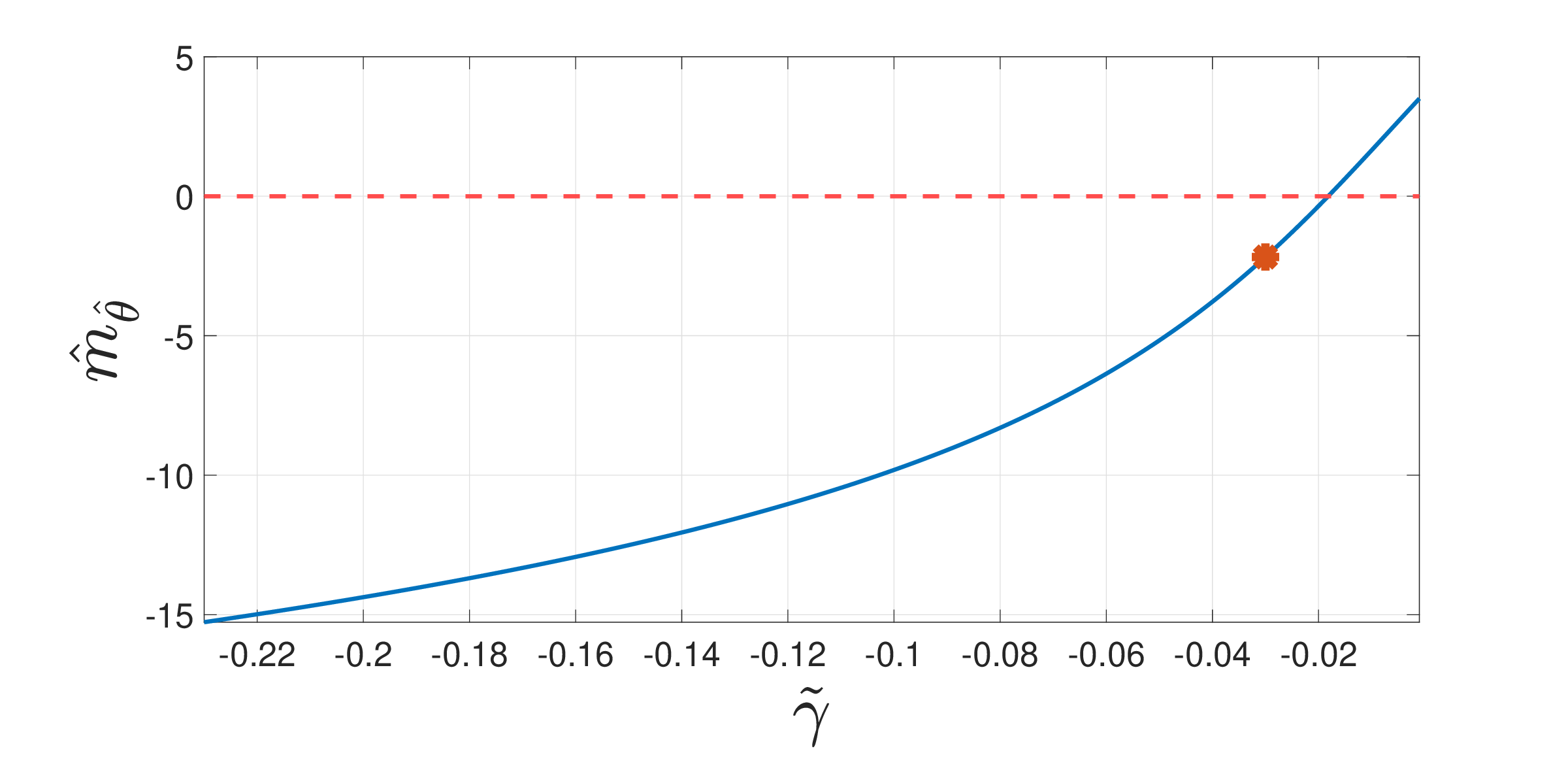}
    \caption{Sensitivity analysis for the optimal multiplier with respect to the jump-diffusion model with constant negative jump size. The blue line is the value of the multiplier for the different values of the parameter indicated under each panel and the red dot corresponds to the value of the multiplier under the parameter configuration in Table \ref{tab:cons_jump_size}.}
    \label{fig:constant}
    \end{figure}
    We have also represented the value function with respect to time to maturity and the initial cushion in {left chart of} Figure \ref{fig:Value_fun_constant_and_traj_constant}. {In particular, it is increasing with respect to both time and the initial value of the cushion and concave}. Most interestingly, we have numerically verified the theoretical result that the gap risk does not occur.
    Indeed, {in right chart of} Figure \ref{fig:Value_fun_constant_and_traj_constant}  we have plotted the dynamic portfolio values $q_0(t)$, $q_{99}(t)$ between $\mathbb{P}( V_{t}\le q_0(t))=0$ and $\mathbb{P}( V_{t}\le q_{99}(t))=0.99$ for every $t \le T$. In particular, the zero quantile $q_0(t)$ is always above the floor value, as indicated in the red dashed line. 
    \begin{figure}[th!]
    \centering
    \includegraphics[width=0.49\linewidth]{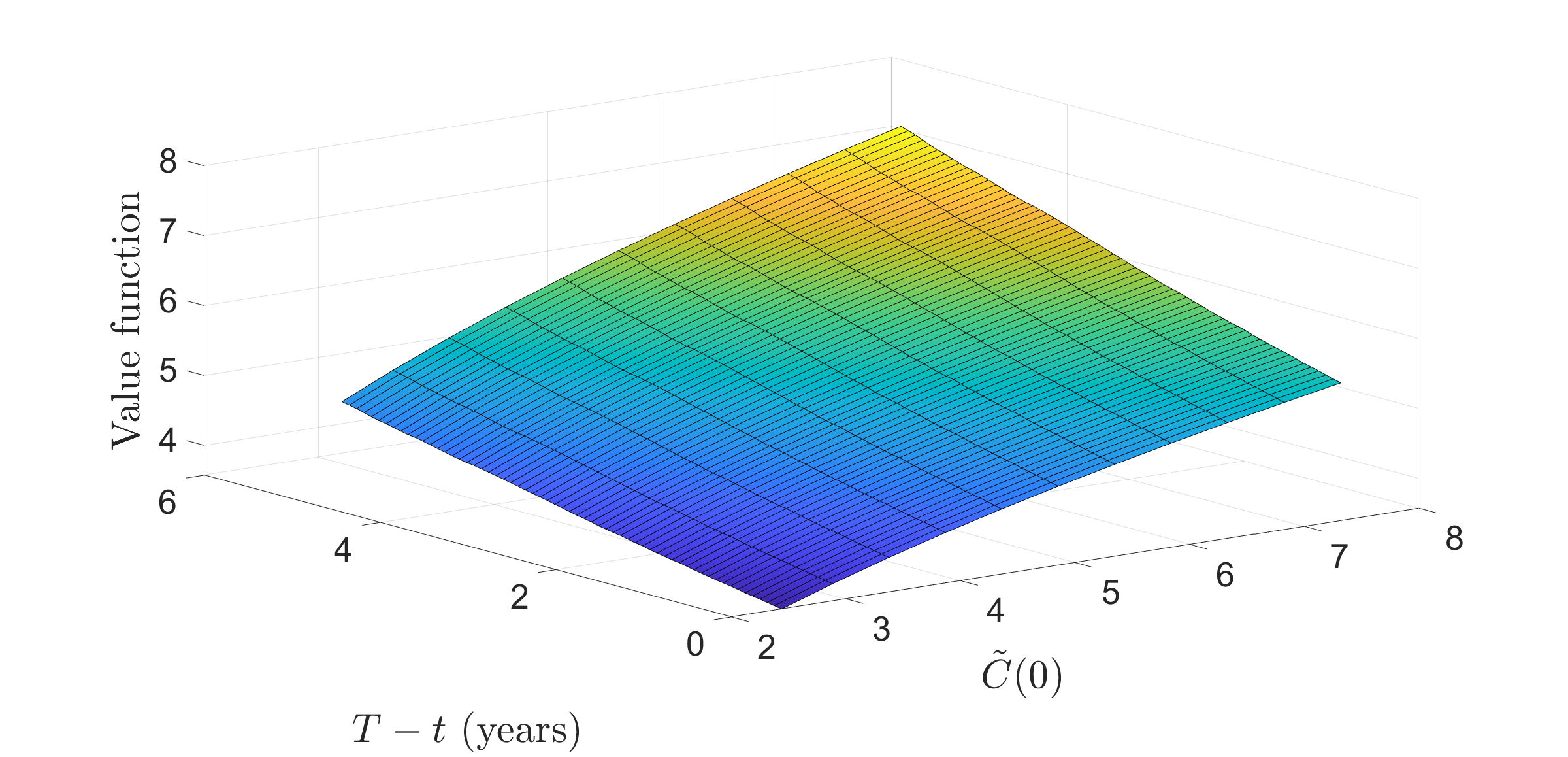}
    \hfill 
    \includegraphics[width=0.49\linewidth]{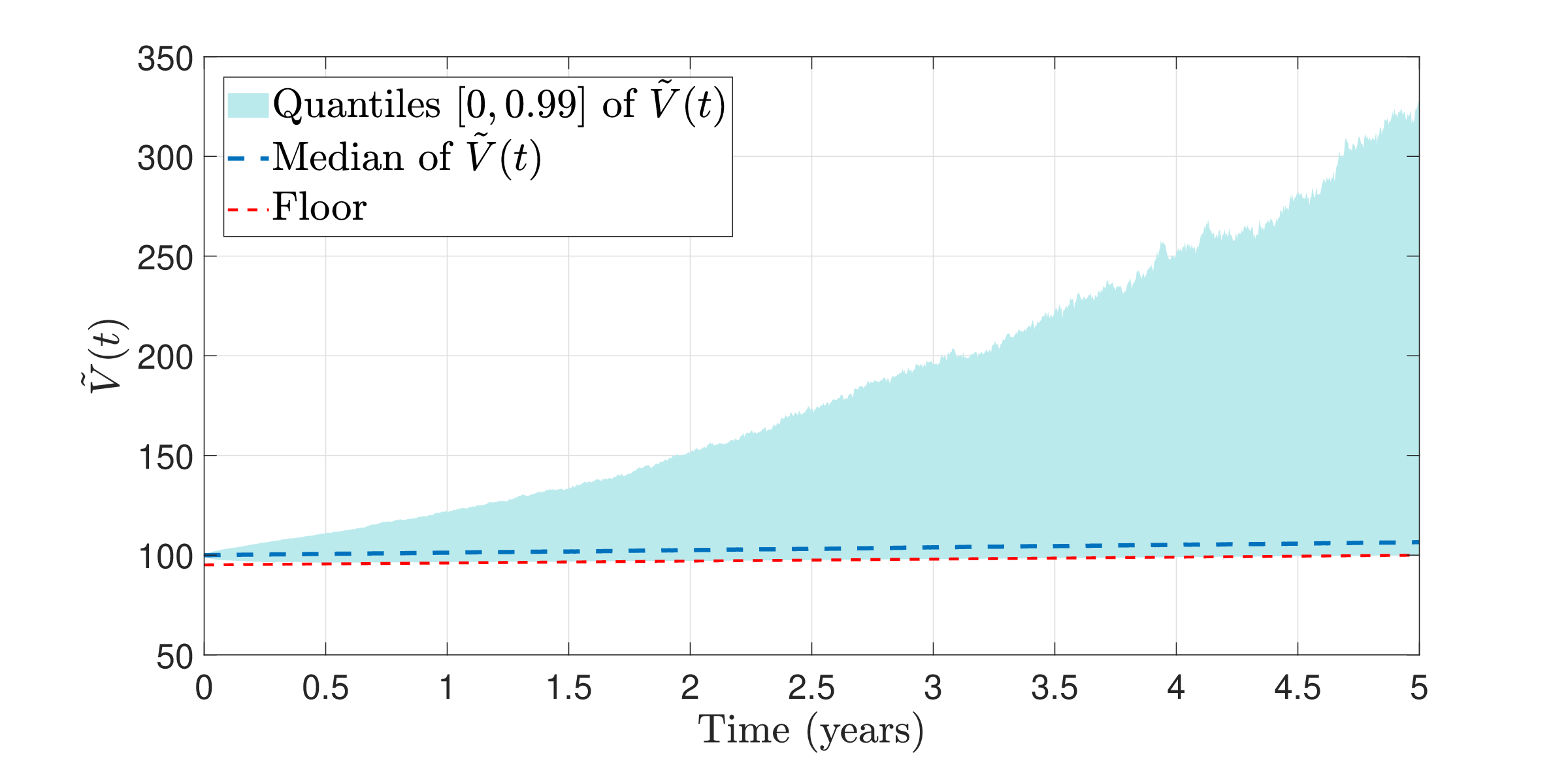}
    \caption{Value function for the constant jump size, with respect to time to maturity and initial cushion (left chart). Median (dotted line) and extreme scenarios for the portfolio value in case of constant jump size. The shaded area represents the portfolio values between the zero and the $99\%$ quantiles. The red dashed line represents the level of the floor which is always below the zero quantile (right chart).}
    \label{fig:Value_fun_constant_and_traj_constant}
    \end{figure}
    We finally perform a numerical exercise to test the performance of the multiplier $\hat m^{\hat{\bm{\theta}}}$, computed for the constant negative jump size against small variations of $\tilde \gamma$. For the test we take the parameters as in Table \ref{tab:cons_jump_size}, jump size  $\gamma=\tilde \gamma + \sigma_{\tilde \gamma} \epsilon$, where $\epsilon$ is a standard Gaussian random variable, and compute $\mathbb{P}(C^{\hat{m}^{\hat {\bm{\theta}}}}_T(\gamma)\le 0)$.  The values, reported in Table \ref{tab:loss}, show that for small errors in the estimation of $\tilde\gamma$, although sub-optimal, the multiplier performs quite well with zero or very small probability of a loss at maturity. The probability of gap risk increases rapidly for larger errors in the estimation of jump sizes.
    \begin{table}[th!]
    \centering
    \begin{tabular}{c|ccccc} 
    \Xhline{1.5pt}
    & $\sigma_{\tilde\gamma}=0.05$ & $\sigma_{\tilde\gamma}=0.075$ & $\sigma_{\tilde\gamma}=0.1$  & $\sigma_{\tilde\gamma}=0.125$ & $\sigma_{\tilde\gamma}=0.15$  \\
    \hline
    $T=1$  & $0$      & $0$       & $0.0003$ & $0.0052$  & $0.0374$ \\
    $T=5$  & $0$      & $0$       & $0.0015$ & $0.0293$  & $0.1471$ \\
    $T=10$ & $0$      & $0$       & $0.0024$ & $0.0565$  & $0.2493$\\
    \Xhline{1.5pt}
    \end{tabular}
    \caption{Probabilities of gap risk for different errors in the estimation of $\tilde \gamma$. The multiplier is computed for the constant jump size, with guaranteed amount at different maturities. All other parameters are set as in Table \ref{tab:cons_jump_size}.}
    \label{tab:loss}
    \end{table} 
    \subsection{Optimal multiplier under Kou's and Merton's models}
    Kou's model assumes the following dynamics for the underlying risky asset $S$
    \begin{equation}\label{eq:Kou_model}
    \frac{\de S_t}{S_{t^-}}=\mu\de t+\sigma\de W_t+\de\left[\sum_{j=1}^{N(t)}\left(V_j-1\right)\right],\quad t\in[0,T],
    \end{equation}
    where $N$ is an homogeneous Poisson process with intensity rate $\lambda>0$, and $\left\lbrace V_1,V_2,\dots\right\rbrace$ is a sequence of independent and identically distributed random variables such that, for every $j$, $\log\left(V_j\right)$ is subject to the following asymmetric double exponential density function
    \begin{equation}\label{eq:double_exp_density_function}
    f(y)=p\eta_{+}e^{-\eta_{+}y}\mathds{1}_{y\geq 0}+\left(1-p\right)\eta_{-}e^{\eta_{-}y}\mathds{1}_{y<0},
    \end{equation}
    where $\eta_{+}>1$, is the parameter governing the severity of upward jumps, $\eta_{-}>0$ is the parameters governing the amplitude of downward jumps, and $p\in(0,1)$ is the probability of an upward jump to occur. By applying the result in Proposition \ref{prop:existence_uniqueness}, we derive the following condition for the existence and uniqueness of the optimal multiplier that maximizes the expected utility of terminal cushion.  
    \begin{cor}\label{cor:KOU_MODEL}
    Assume that the stock price process follows the Kou's model stated in equation \eqref{eq:Kou_model}. If the model parameters satisfy
    \begin{align}
    \label{eq:g_0_Kou}&\mu-r-\dfrac{p\lambda}{1-\eta_{+}}-\frac{\left(1-p\right)\lambda}{1+\eta_{-}}\geq 0,\\
    \label{eq:d_1_Kou}&\mu-r-\delta_1\sigma^2-\frac{p\eta_{+}\lambda}{\left(1-\delta_1-\eta_{+}\right)\left(\delta_1+\eta_{+}\right)}+\frac{(1-p)\eta_{-}\lambda}{\left(1-\delta_1+\eta_{-}\right)\left(\delta_1-\eta_{-}\right)}\leq 0,
    \end{align}
    then there exists a unique optimal multiplier $\hat{m}^{\hat{\bm{\theta}}}\in\left[0,1\right]$ that can be found by solving
    \begin{equation}
    \mu-r-\delta_1\sigma^2 m+\lambda\int_{\mathbb{R}\setminus\left\lbrace0\right\rbrace}\frac{e^y-1}{\left[1+\left(e^y-1\right)m\right]^{\delta_1}}\left[p\eta_{+}e^{-\eta_{+}y}\mathds{1}_{y\geq 0}+\left(1-p\right)\eta_{-}e^{\eta_{-}y}\mathds{1}_{y<0}\right]\de  y=0.
    \end{equation}
    \end{cor}
    \begin{proof}
    See Appendix \ref{app:C_2}.
    \end{proof}
    As in the previous example we consider the set of parameters in Table \ref{tab:Kou_model_par}, and change their values one by one. 
    Figure \ref{fig:Kou_sensitivity} shows results consistent with the case of constant jump size. 
    \begin{table}[H]
    \centering
    \begin{tabular}{cccccccc} 
    \hline
    $\mu-r$  & $\sigma$ & $\lambda$ & $p$     & $\eta_+$ & $\eta_{-}$ & $\delta_1$ & $\hat{m}^{\hat{\bm{\theta}}}$ \\ 
    \hline
    $0.24$   & $0.26$ & $20$      & $0.72$ & $64.94$   & $49.02$ & $0.60$ & $0.77$\\
    \hline
    \end{tabular}
    \caption{Parameters for the Kou model.}\label{tab:Kou_model_par}
    \end{table}
    \begin{figure}[th!]
    \centering
    \includegraphics[width=0.32\linewidth]{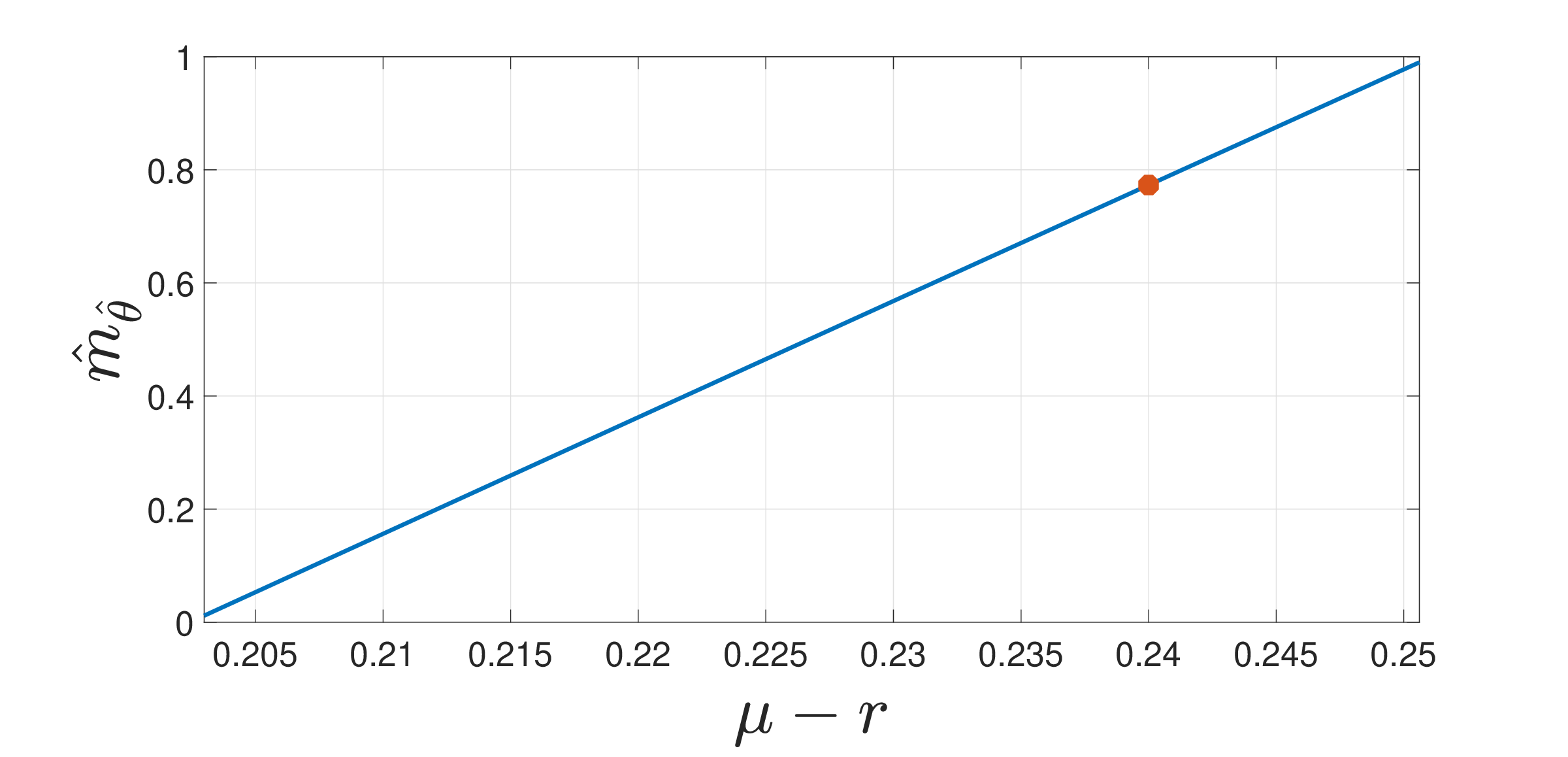}
    \hfill % Spazio orizzontale tra le immagini
    \includegraphics[width=0.32\linewidth]{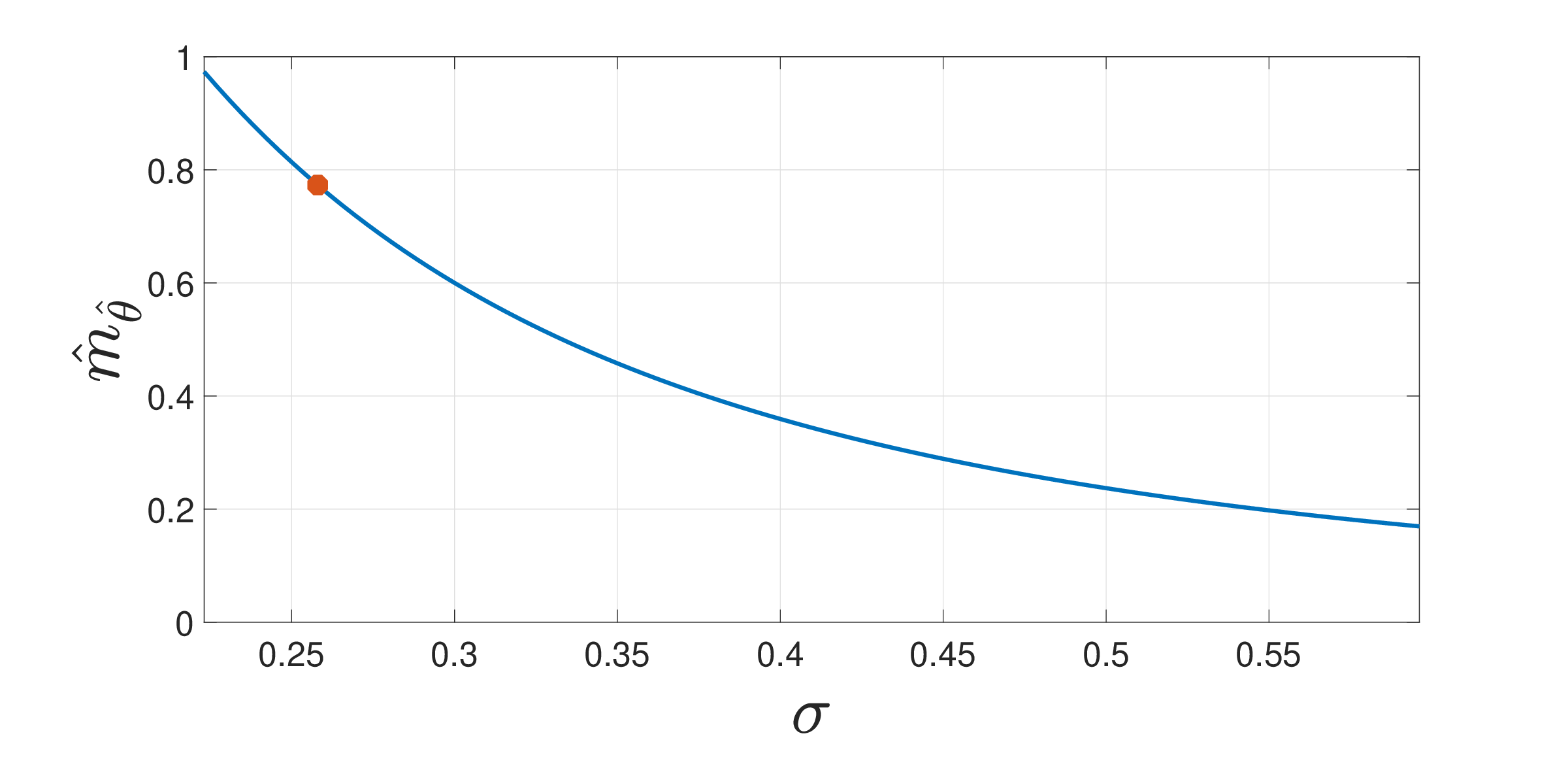}
    \hfill
    \includegraphics[width=0.32\linewidth]{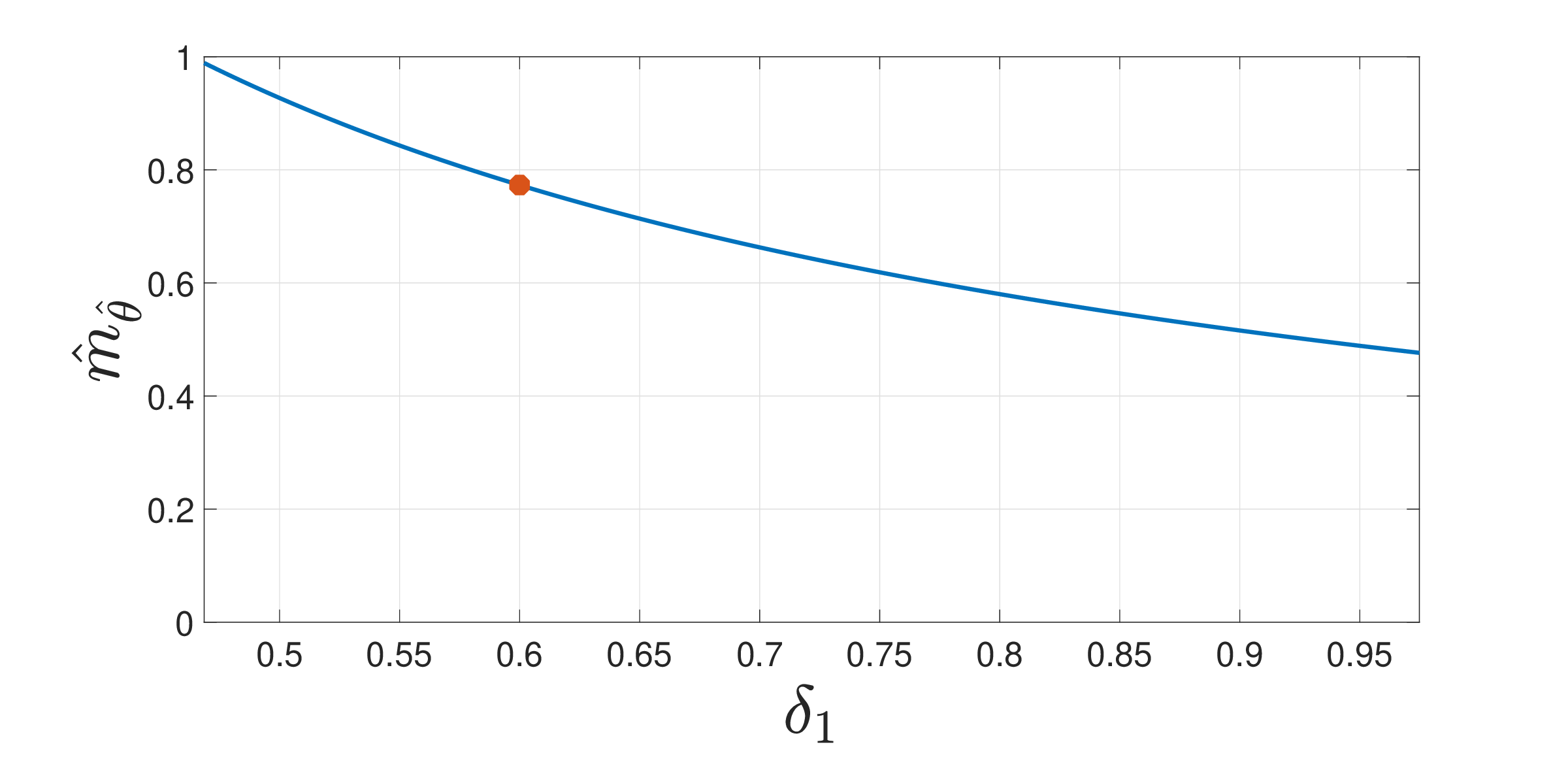}
    \vspace{0.5cm} % Spazio verticale tra le righe di figure
    \includegraphics[width=0.32\linewidth]{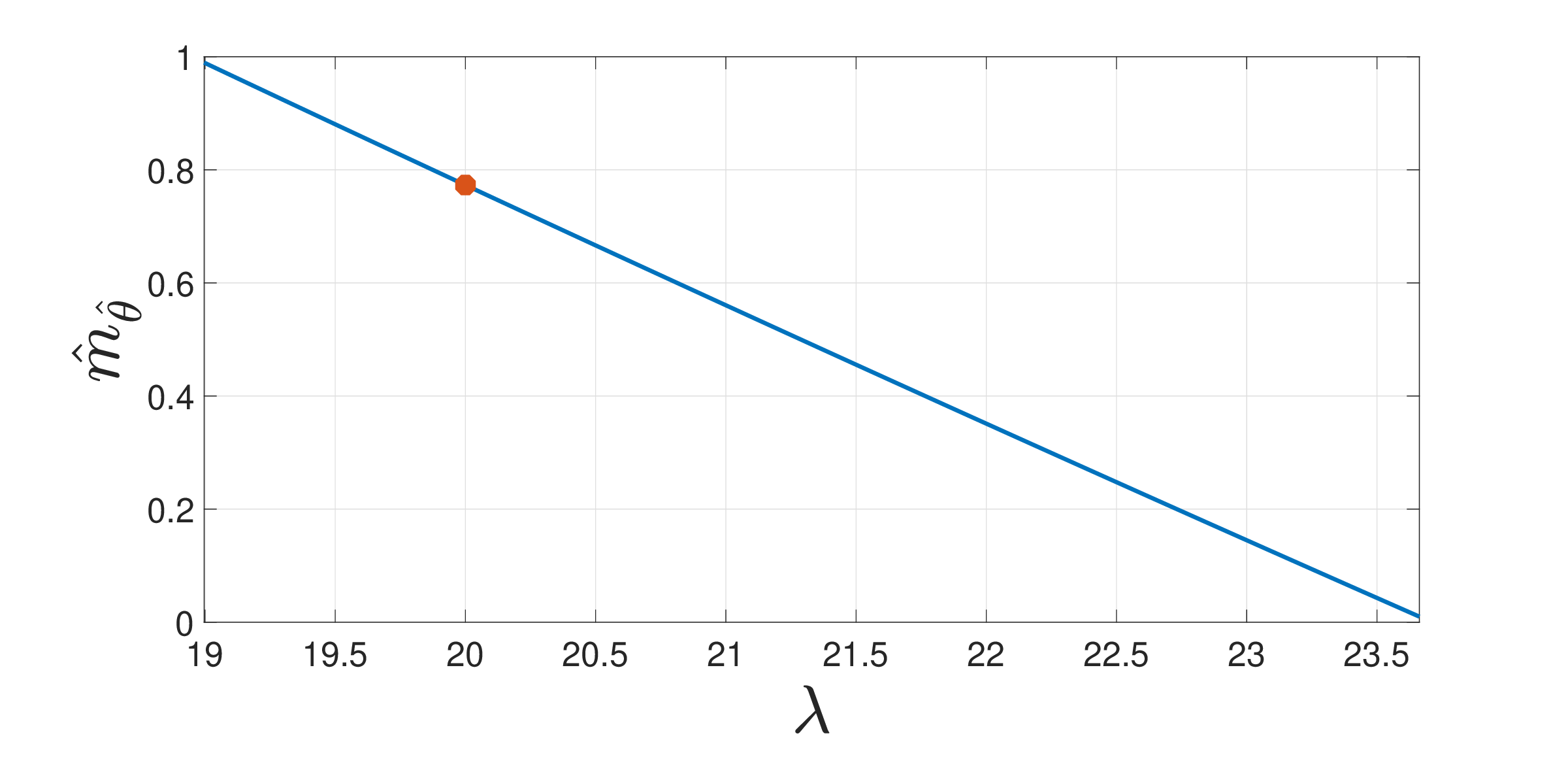}
    \hfill
    \includegraphics[width=0.32\linewidth]{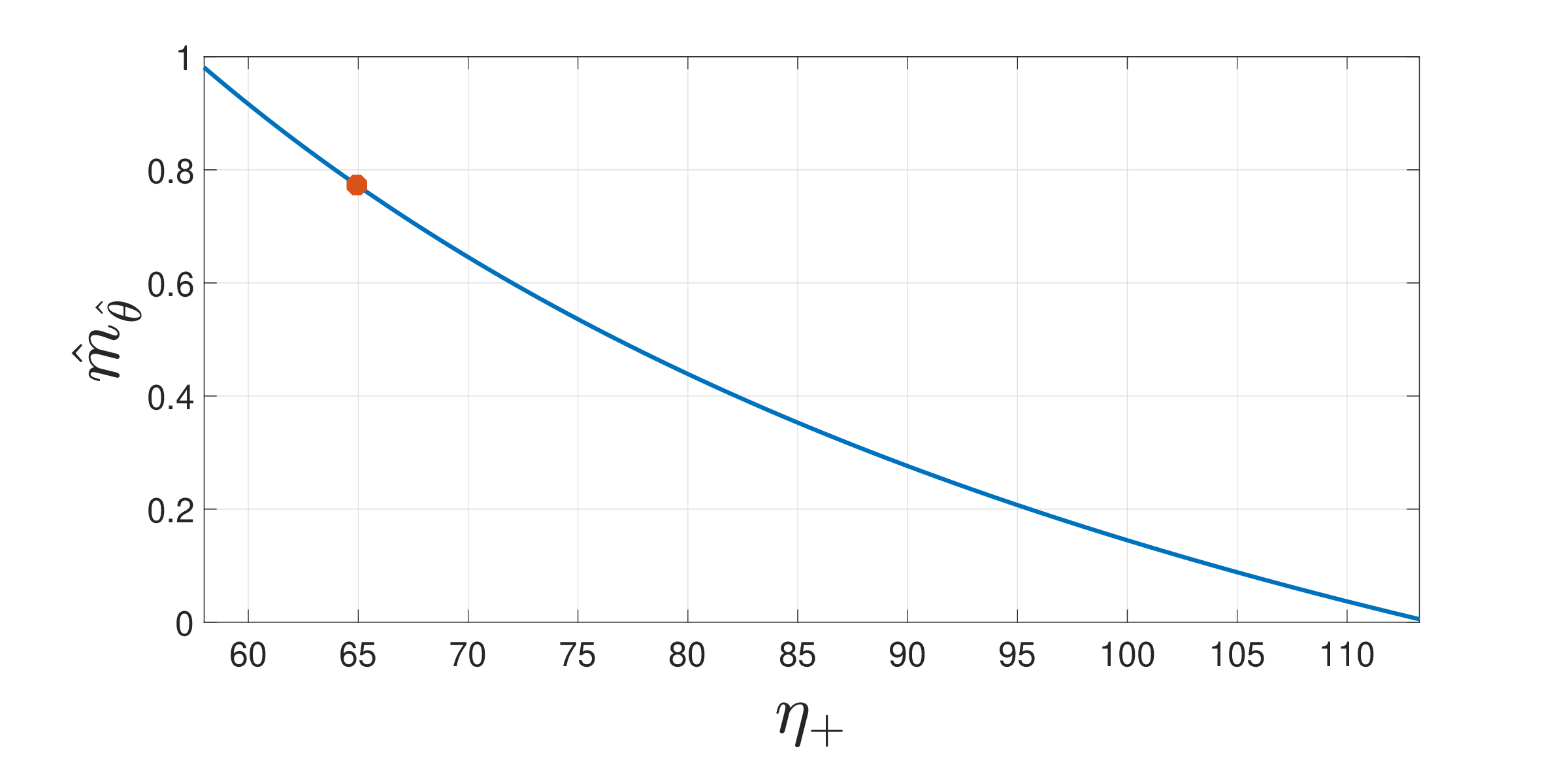}
    \hfill
    \includegraphics[width=0.32\linewidth]{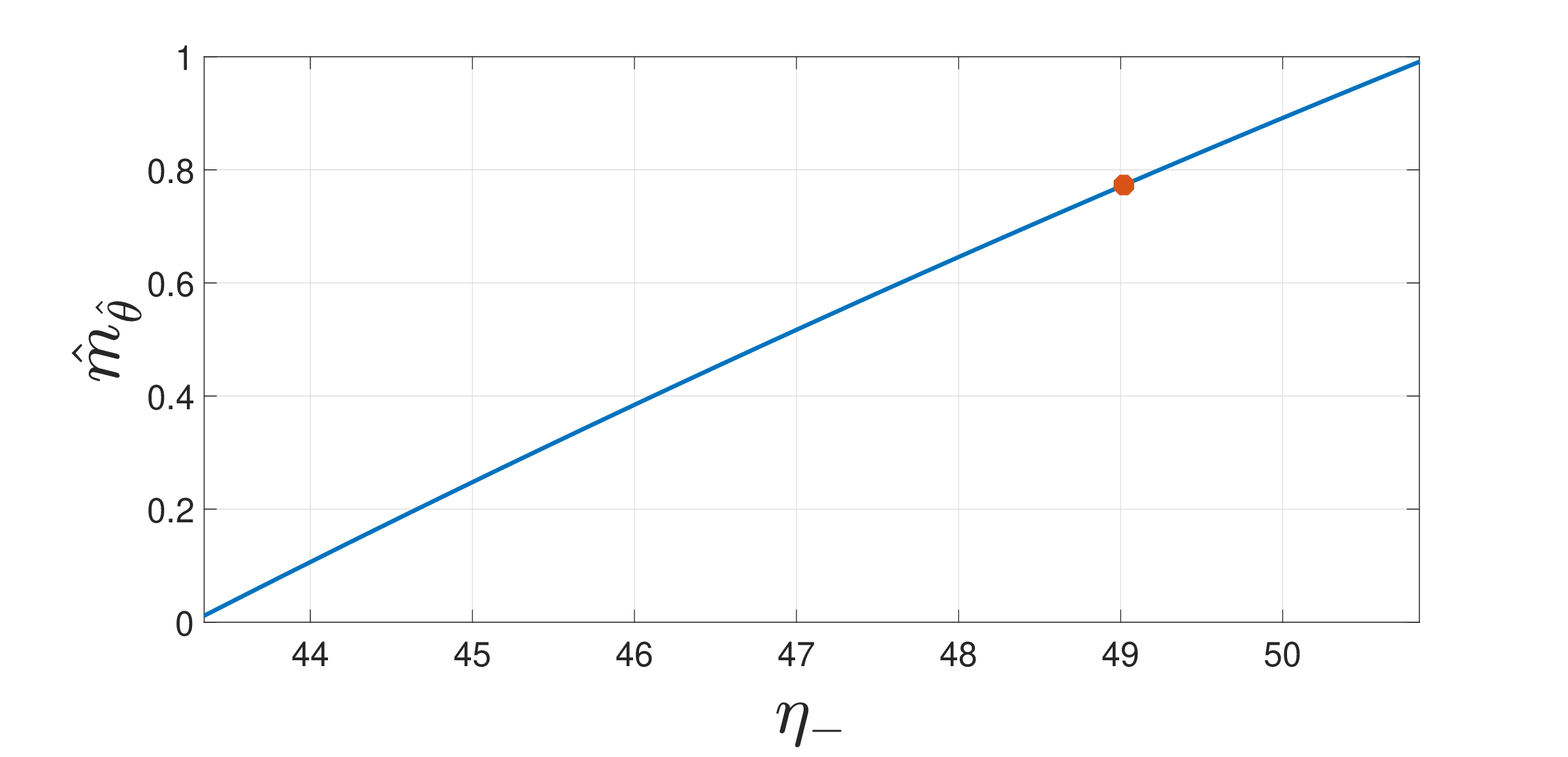}
    \vspace{0.5cm} % Spazio verticale tra le righe di figure
    \includegraphics[width=0.32\linewidth]{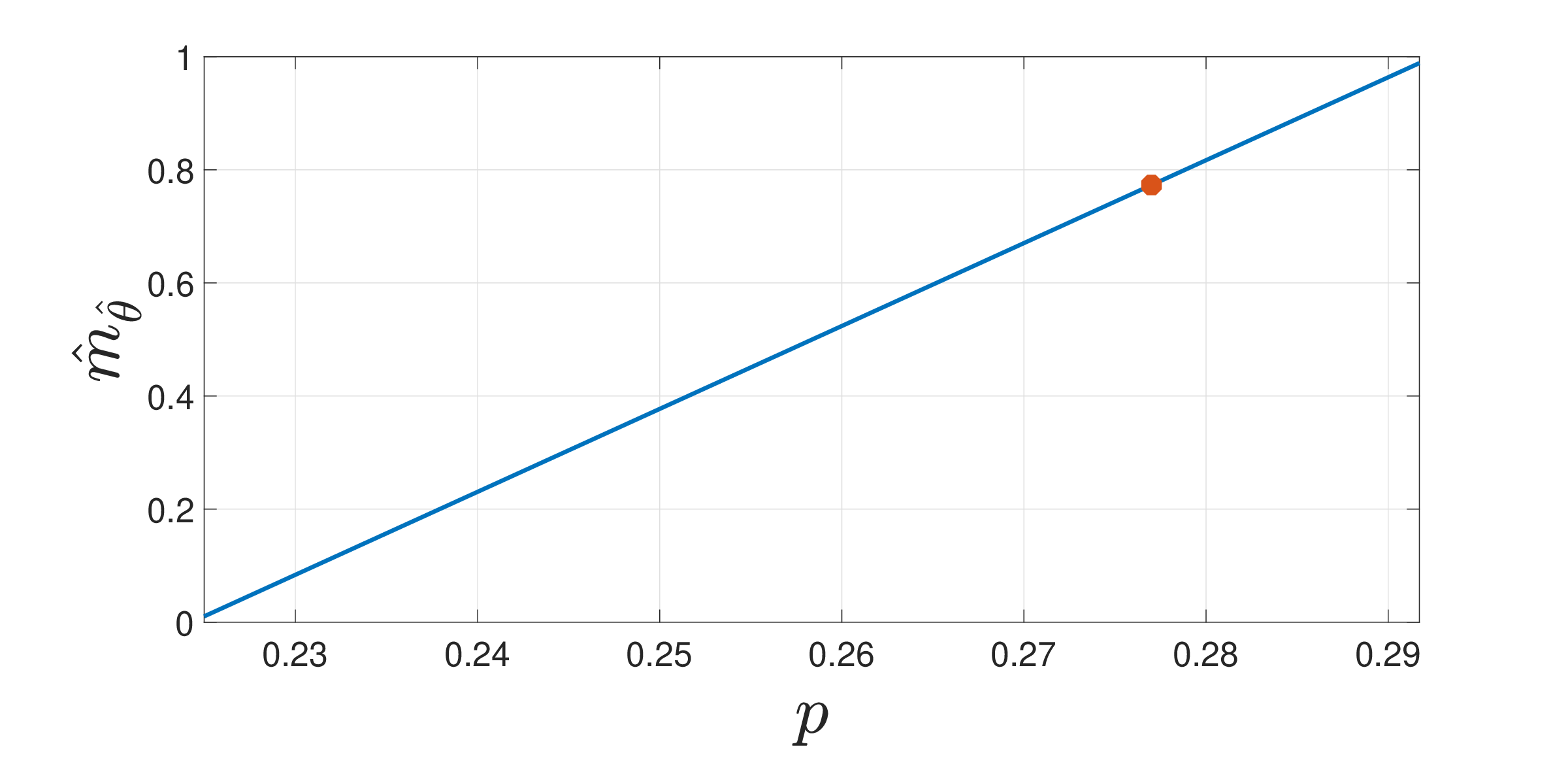}
    \caption{Sensitivity analysis for the optimal multiplier with respect to Kou's model parameters. The blue line is the value of the multiplier for the different values of the parameter indicated under each panel and the red dot corresponds to the value of the multiplier under the parameter configuration in Table \ref{tab:Kou_model_par}.}
    \label{fig:Kou_sensitivity}
    \end{figure}
    We stress that when $\eta_{-}$ increases, the expected size of negative jumps decreases. That would mean smaller average downward moves of the risky asset price, suggesting a larger exposure to the risky asset may be advantageous. The situation reverts when we look at $\eta_+$. Indeed, when the parameter governing the severity of positive jumps increases, upward moves become less pronounced, leading to a smaller exposure. What is interesting in this case is that the optimal multiplier assumes value in $[0,1]$, which, in virtue of the relation $\pi_t  V_t=m_{t}  C_{t}$, translates into an investment strategy where short-selling and borrowing from the money market account are not allowed (here we also used the fact that $ C_{t}< V_t$). The value function for the optimization problem in Kou's specification shares the same features as the one depicted in Figure \ref{fig:Value_fun_constant_and_traj_constant} for the constant jump size case. Hence, for the sake of brevity, we have omitted it here. Figure \ref{fig:traj_Kou} shows the extreme scenarios for the portfolio values, and it confirms that the optimal investment strategy allows the portfolio value process to stay well above the floor, indicated with the dotted red line.
    \begin{figure}[th!]
    \centering
    \includegraphics[width=0.60\linewidth]{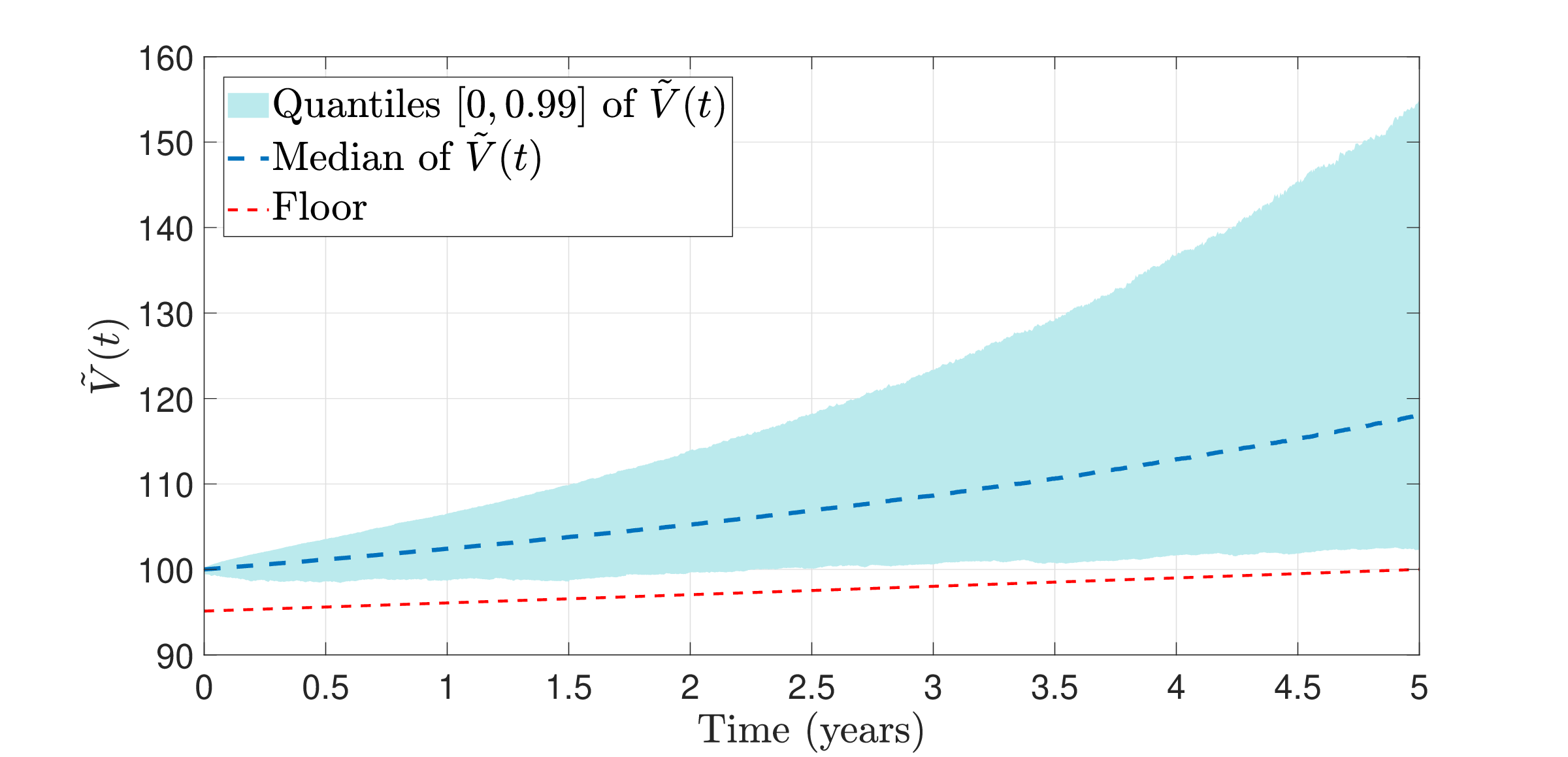}
    \caption{Median (dotted line) and extreme scenarios for the portfolio value in case of Kou's model. The shaded area represents the portfolio values between the zero and the $99\%$ quantiles. The red dashed line represents the level of the floor which is always below the zero quantile.}
    \label{fig:traj_Kou}
    \end{figure}
    Under Merton's model the price process of the underlying risky asset is still defined as in equation \eqref{eq:Kou_model}, with a different assumption on the jump size distribution. In particular, jump amplitudes $\left\lbrace V_1,V_2,\dots\right\rbrace$ are independent identically distributed random variables such that, for all $j$,  $\log(V_j)$ has normal density function  
    \begin{equation}
    f(y)=\frac{1}{\sqrt{2\pi\sigma_J^2}}\exp\left\lbrace-\dfrac{\left(y-\mu_J\right)^2}{2\sigma_J^2}\right\rbrace,
    \end{equation}
    for constant parameters $\mu_J\in\mathbb{R}$, and $\sigma_J>0$ representing the mean and the standard deviation of jump size distribution, respectively. 
    For this model, the parameters' conditions for the optimal multiplier to exist and be unique are as follows. 
    \begin{cor}\label{eq:corollary_Merton} 
    Assume that the stock price process follows the Merton model. If the model parameters satisfy
    \begin{align}
&\mu+\lambda\left(\exp\left\lbrace\mu_J+\frac{\sigma_J^2}{2}\right\rbrace-1\right)\geq 0,\\
    &\mu-\sigma^2\delta_1+\lambda\left(\exp\left\lbrace(1-\delta_1)\mu_J+\frac{(1-\delta_1)^2\sigma_J^2}{2}\right\rbrace-\exp\left\lbrace-\delta_1\mu_J+\frac{\delta_1^2\sigma_J^2}{2}\right\rbrace\right)\leq 0,
    \end{align}
    then there exists a unique optimal multiplier $\hat{m}^{\hat{\bm{\theta}}}\in[0,1]$ that solves
    \begin{equation}
    \mu-r-\delta_1\sigma^2 m+\lambda\int_{\mathbb{R}\setminus\left\lbrace0\right\rbrace}\frac{e^y-1}{\left[1+\left(e^y-1\right)m\right]^{\delta_1}}\frac{1}{\sqrt{2\pi\sigma_J^2}}\exp\left\lbrace-\dfrac{\left(y-\mu_J\right)^2}{2\sigma_J^2}\right\rbrace\de  y=0.
    \end{equation}
    \end{cor}
    \begin{proof}
    The proof replicates the lines of that of Corollary \ref{cor:KOU_MODEL}.
    \end{proof}
    We now fix the parameter according to the following values and perform a sensitivity analysis by varying each of them.
    \begin{table}[H]
    \centering
    \begin{tabular}{ccccccc} 
    \hline
    $\mu-r$ & $\sigma$ & $\lambda$ & $\mu_J$  & $\sigma_J$ & $\delta_1$ & $\hat{m}^{\hat{\bm{\theta}}}$\\
    \hline
    $0.09$ & $0.35$  & $20$     & $-0.01$ & $0.15$ & $0.60$ & $0.22$\\
    \hline
    \end{tabular}
    \caption{Parameters for the Merton model.}
    \label{tab:Merton_model_par}
    \end{table} 
    In Figure \ref{fig:MERTON_sens} we analyze the behaviour of the optimal multiplier with respect to  excess return, volatility, frequency of jumps in the underlying risky asset, and the investor's risk aversion. The results are comparable to those in the previous cases (i.e. constant jump sizes and the Kou model). 
    The optimal multiplier increases with the parameter $\mu_J$, which indicates higher exposure to the risky asset as the expected value of jumps sizes increases, for fixed volatility $\sigma_J=0.15$. 
    The optimal multiplier also increases with respect to standard deviation of the jump size distribution. 
    Indeed, when $\sigma_J$ is small, jump sizes tend to cluster around negative values. Conversely, larger values of $\sigma_J$ lead to dispersion of jump sizes across a wider range, encompassing positive values. 
    As for Kou's model, the optimal multiplier assumes a value in $[0,1]$, implying that our new investment strategy does not allow for short-selling or borrowing from the bank account. Since the value function has the same features as in the previous two specifications of jump size distributions, we have not reported it.    
    \begin{figure}[H]
    \centering
    \includegraphics[width=0.32\linewidth]{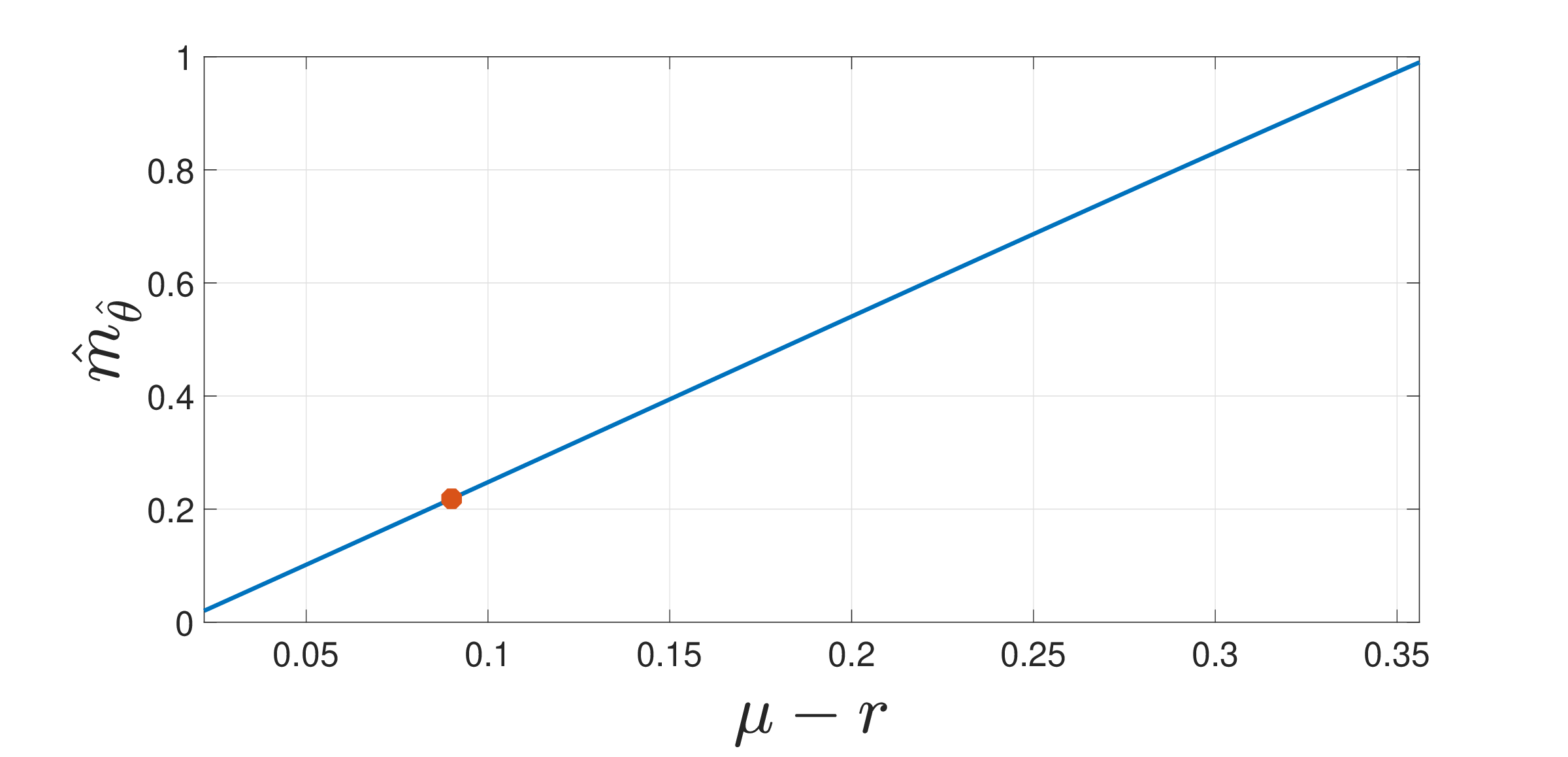}
    \hfill % Spazio orizzontale tra le immagini
    \includegraphics[width=0.32\linewidth]{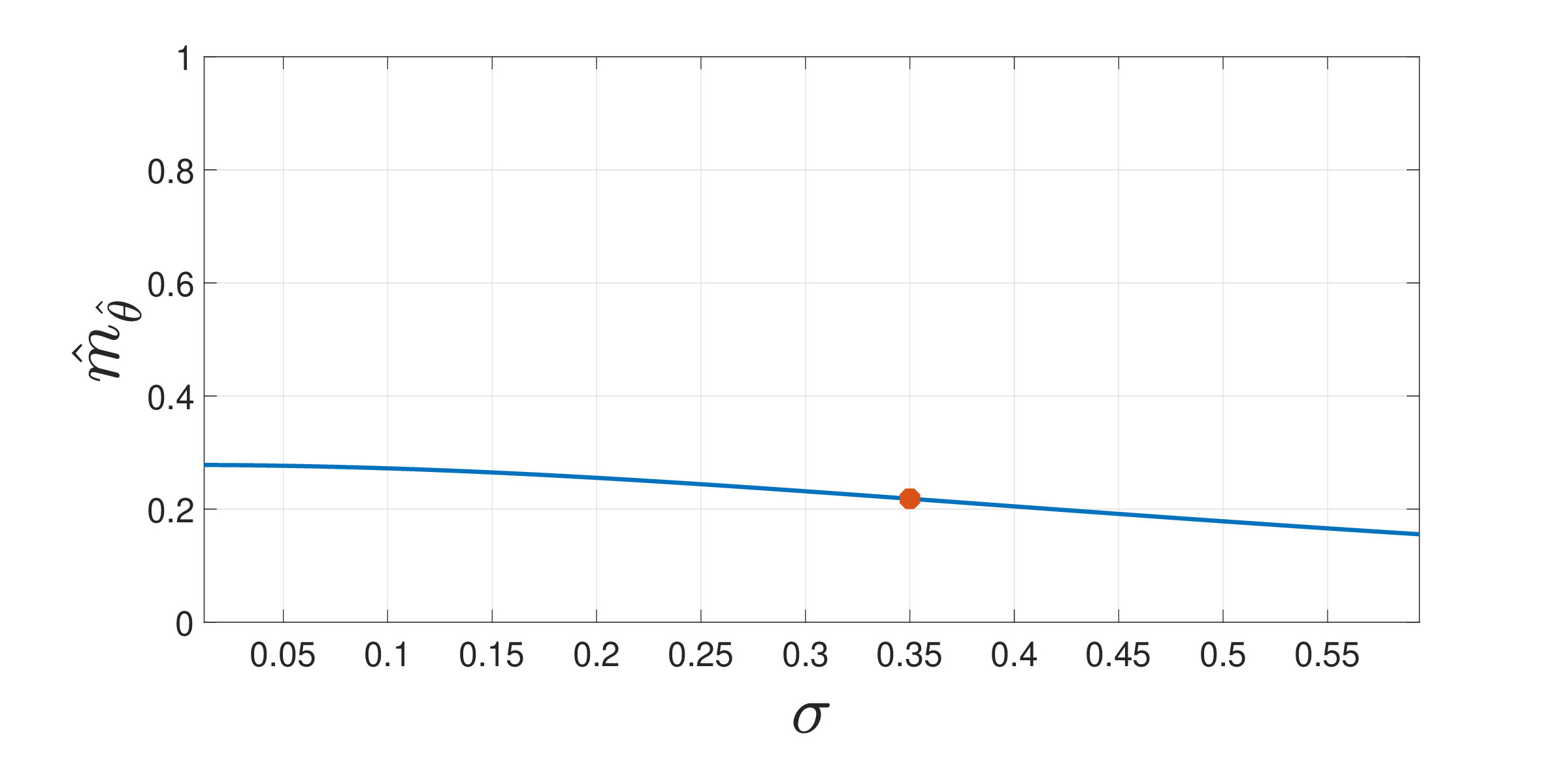}
    \hfill
    \includegraphics[width=0.32\linewidth]{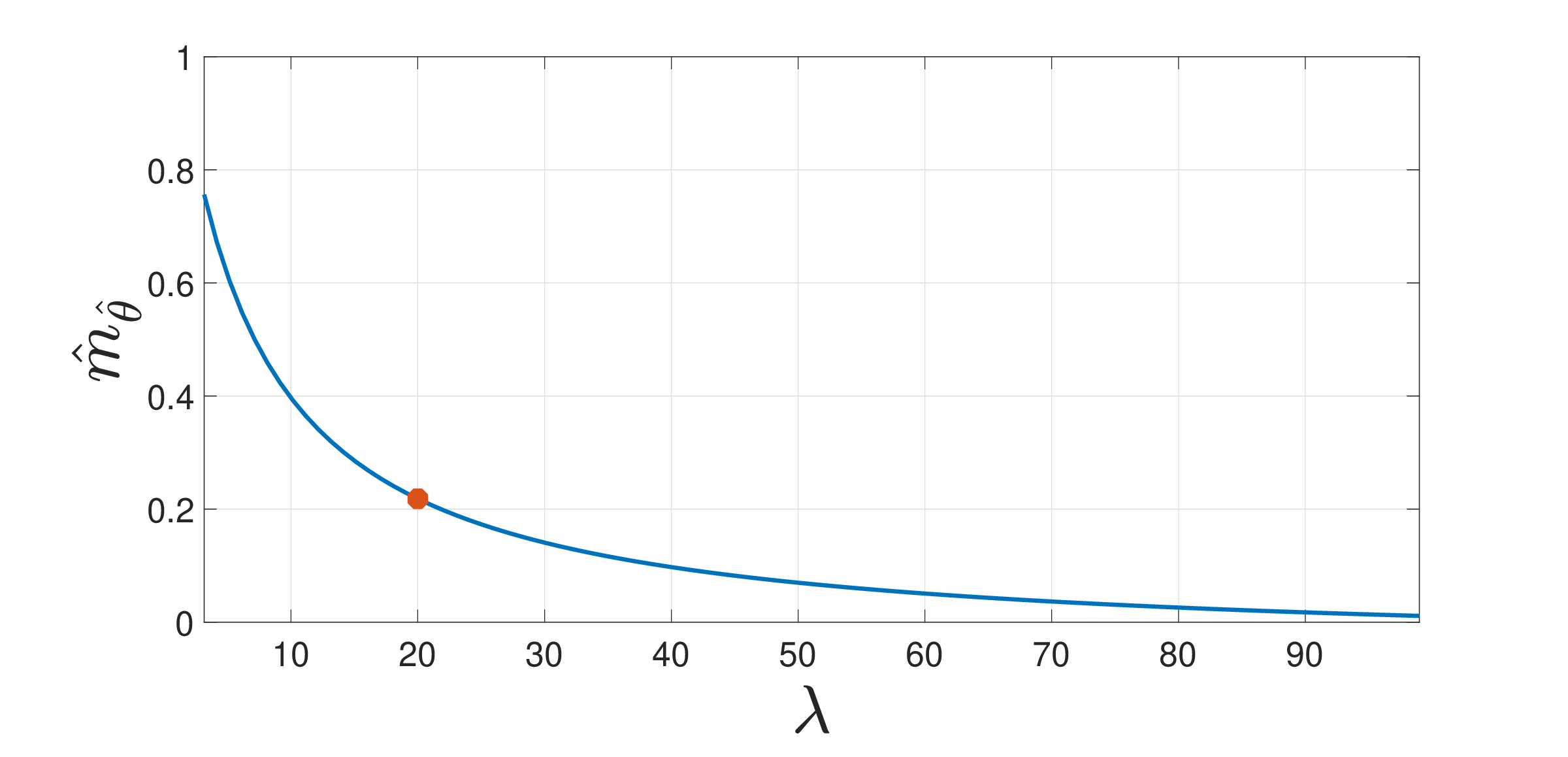}
    \vspace{0.5cm} % Spazio verticale tra le righe di figure
    \includegraphics[width=0.32\linewidth]{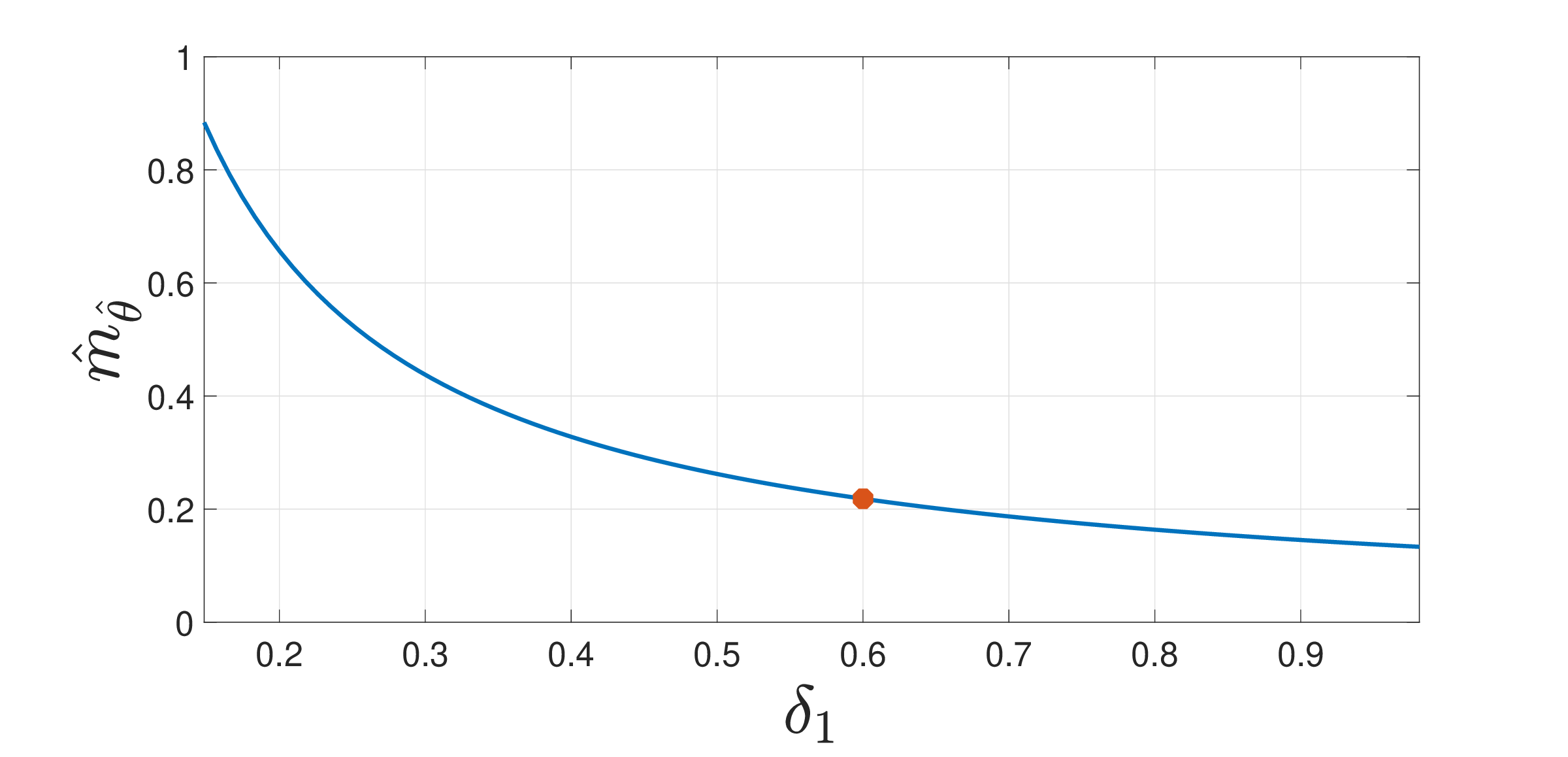}
    \hfill
    \includegraphics[width=0.32\linewidth]{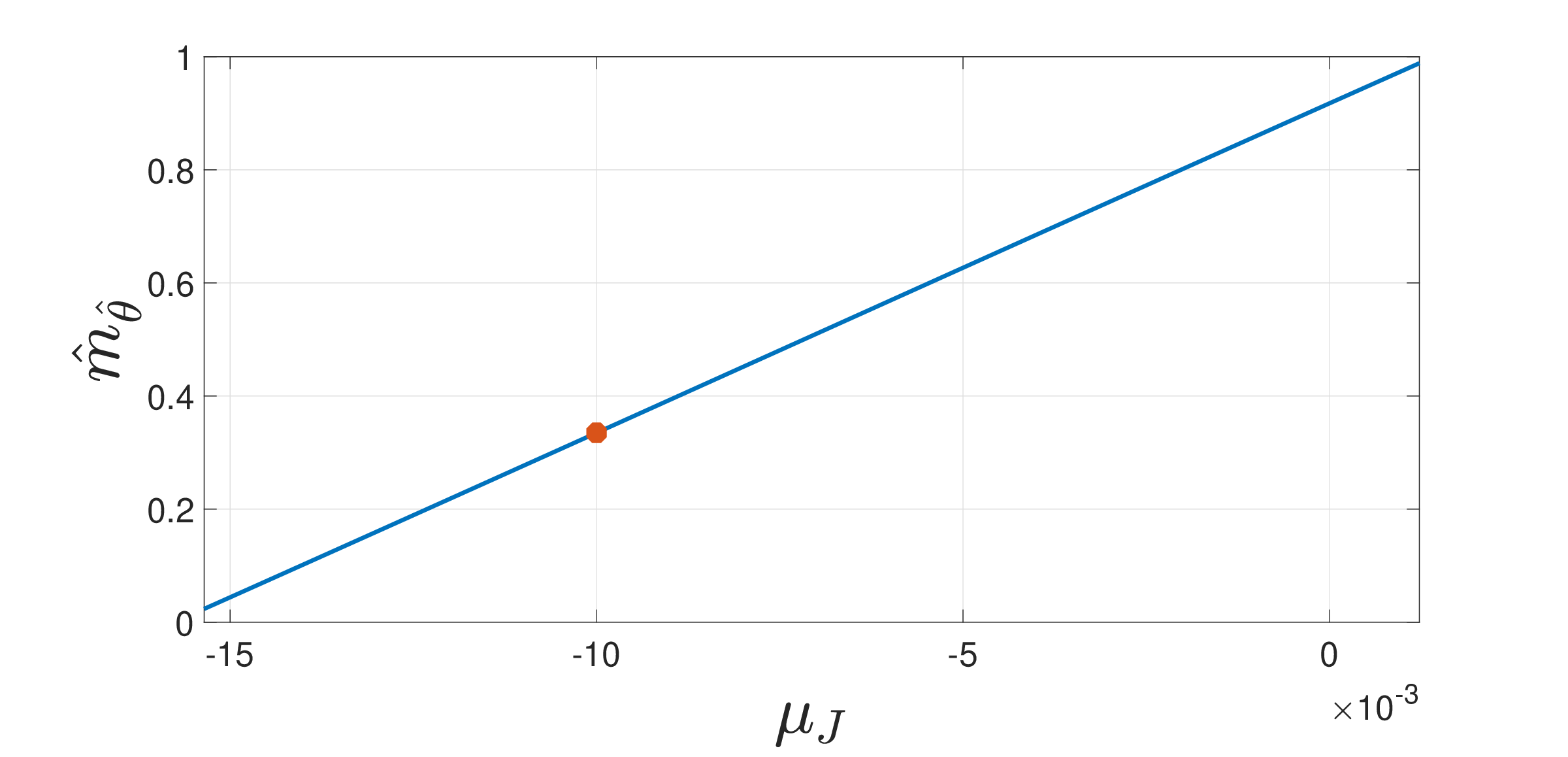}
    \hfill
    \includegraphics[width=0.32\linewidth]{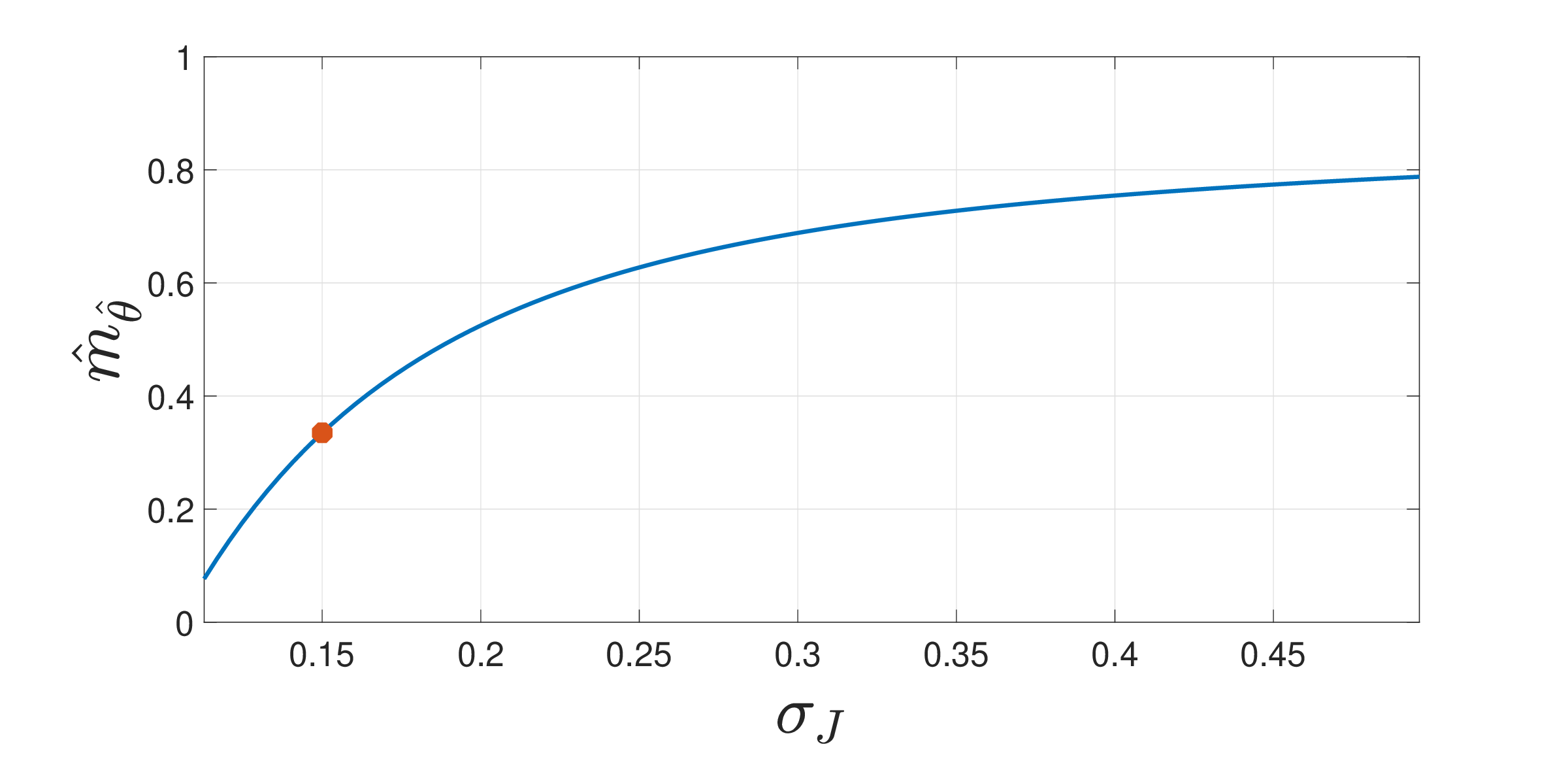}
    \caption{Sensitivity analysis for the optimal multiplier with respect to Merton's model parameters. The blue line is the value of the multiplier for the different values of the parameter indicated under each panel and the red dot corresponds to the value of the multiplier under the parameter configuration in Table \ref{tab:Merton_model_par}.}
    \label{fig:MERTON_sens}
    \end{figure}
    Furthermore, for comprehensive analysis, we display the extreme scenarios of portfolio values in Figure \ref{fig:traj_merton}. These scenarios highlight that the optimal investment strategy maintains the portfolio value above the floor, eliminating gap risk throughout the entire investment horizon $[0,T]$.
    \begin{figure}[htbp]
    \centering
    \includegraphics[width=0.60\linewidth]{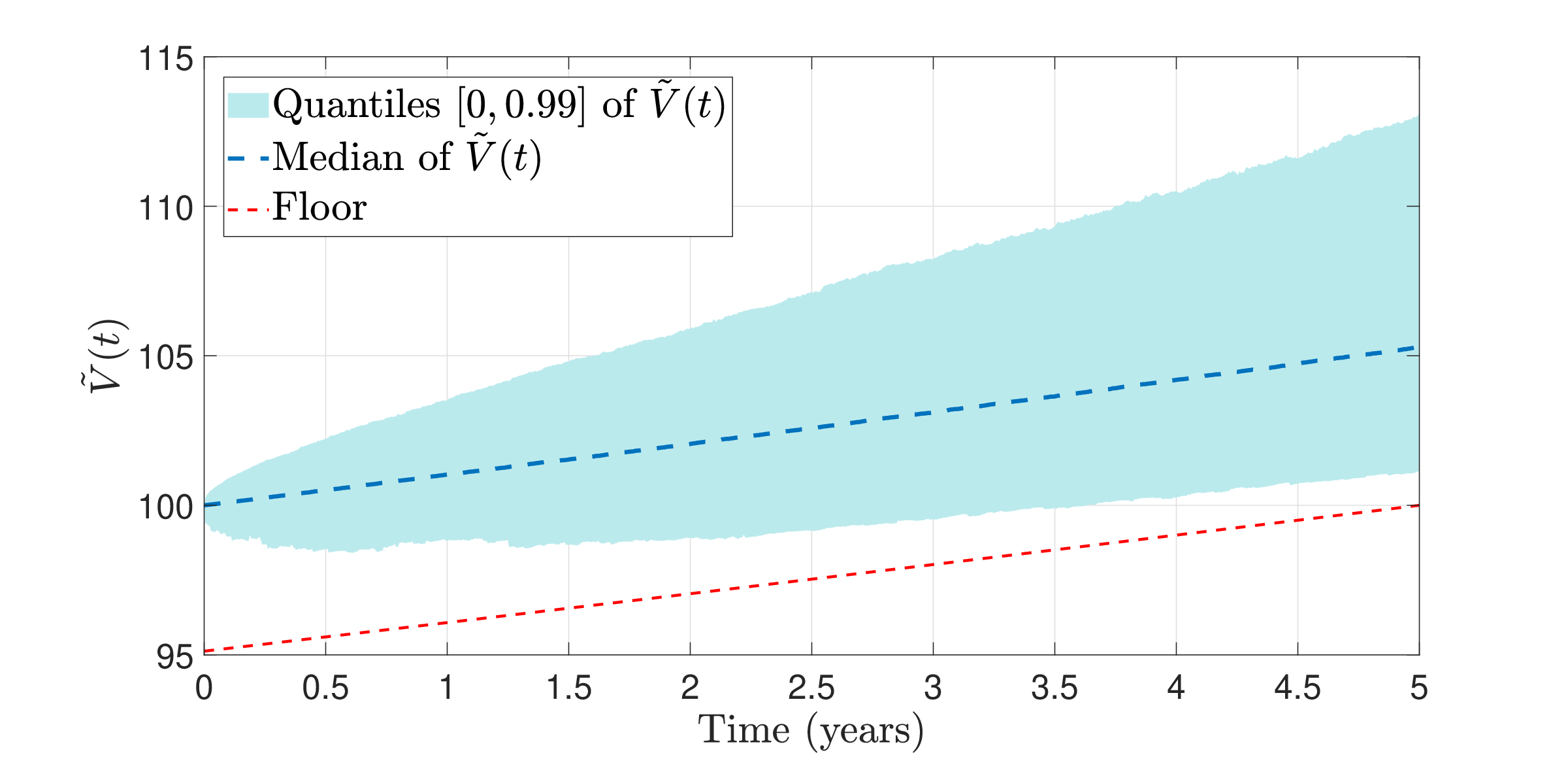}
    \caption{Median (dotted line) and extreme scenarios for the portfolio value in case of Merton's model. The shaded area represents the portfolio values between the zero and the $99\%$ quantiles. The red dashed line represents the level of the floor which is always below the zero quantile.}
    \label{fig:traj_merton}
    \end{figure}
    \section{Conclusions}\label{sect:conclusion}
    PPI strategies are {pivotal in the} asset management industry. In this paper, we have considered the problem of optimal design of a novel version of  PPI strategies, in a mathematical setup where portfolio dynamics may exhibit downward jumps, generated by, e.g., sudden market draw-down. Incorporating such feature in a market model brings gap risk with positive probability, as discussed, for instance, by \cite{cont2009constant}. Hence, a relevant task for a manager or an insurer is to hedge against gap risk in order to be able to fulfil the liabilities at maturity. In this paper, we have attempted to provide a criterion for determining the optimal PPI strategy, enabling to capture market rises after drops of the portfolio below the value of the guarantee. This is done by sustaining equity market participation during the entire trading horizon and by evaluating risk aversion in different ways depending on the relative value of the portfolio and the guaranteed amount. From a technical viewpoint, we have addressed the optimization using a worst-case martingale approach and concavification to account for the features of the utility function. 
    An alternative approach like the dynamic programming based on the Hamilton-Jacobi-Bellman equation, can be applied to solve the stochastic optimization. However, in our setup, it may only provide a characterization of the PPI strategy in a feedback form, in terms of the value function which can only be computed numerically. 
    Our main result, provides a procedure that simultaneously computes the state price density (for the worst-case probability measure) and the optimized multiplier, in terms of the solution of a system of equations. Unfortunately, the solution of the system may not exist or may not be unique for some model parameter configurations. In this case, it is hard to identify the optimal PPI strategy. However, we have discussed examples where the multiplier is given in quasi-closed form, making our results readily applicable from a practical point of view. Interestingly, in such examples, it can be explicitly proved that gap risk does not occur. Our study offers several possibility for extensions. Our next step includes enriching the structure of the market by considering, e.g. a portfolio with several underlying assets subject to common contagion factors. 
\section*{Acknowledgement}
The authors would like to thank Alessandro Doldi for useful comments and discussions. The work of Katia Colaneri has been funded by the European Union - Next Generation EU - Project PRIN 2022 [grant code BEMMLZ] with the title \emph{Stochastic control and games and the role of information}; the work of Immacolata Oliva was supported by Sapienza University of Rome under the project \textit{Stochastic Optimization Problems in Insurance, Finance and Economics} [grant number RG123188B4271AD5]. This work has been completed while Daniele Mancinelli was affiliated with Sapienza University of Rome and partially supported by Sapienza University of Rome under the project \textit{Stochastic Optimization Problems in Insurance, Finance and Economics} [grant number RG123188B4271AD5].
    
    \section*{Appendix}
    \appendix
    \numberwithin{equation}{section}
    \renewcommand{\theequation}{\thesection\arabic{equation}}
    \section{Proof of some technical results of Section \ref{sec:optimization_pb}}\label{app:A}
    \subsection{Proof of Lemma \ref{prop:existence_uniqueness_point_c_hat}}\label{app:A_1}
    Recall that $\tilde{U}(x;G)$ is differentiable for every $x\neq 0$.  Using the condition stated in equation \eqref{eq:tangent_point}, we define the function
    \begin{equation}\label{eq:fun_zero}
    f(x):=
    \begin{cases}
    -\tilde\lambda\frac{\delta_2}{1-\delta_2}\left(-x\right)^{1-\delta_2}-\tilde\lambda\frac{G}{\left(-x\right)^{\delta_2}}+\tilde\lambda\frac{G^{1-\delta_2}}{1-\delta_2},\quad x\in(-G,0),\\
    \frac{\delta_1}{1-\delta_1}x^{1-\delta_1}-\frac{G}{x^{\delta_1}}+\tilde\lambda\frac{G^{1-\delta_2}}{1-\delta_2}, \quad  x\in(0,\infty).
    \end{cases}
    \end{equation}
    We see that $\displaystyle\lim_{x \to -G^+}f(x)=0$, $\displaystyle\lim_{x\to 0^{-}}f(x)=-\infty$, moreover $f'(x)<0$   for $x\in(-G, 0)$. Hence there is no $\hat{c}(G)\in\left(-G,0\right)$ such that $f(\hat{c}(G))=0$. If we consider $x\in\left(0,+\infty\right)$, since $\displaystyle\lim_{x\to 0^{+}}f(x)=-\infty$ and $\displaystyle\lim_{x\to +\infty}f(x)=+\infty$, then there exists $\hat{c}(G)>0$ such that $f(\hat{c}(G))=0$. Moreover, since $f'(x)>0$ for all $x\in(0,+\infty)$, then $\hat{c}(G)$ is also unique. This concludes the proof. 
    \subsection{Proof of Lemma \ref{lem:2}}\label{app:A_2}
    
    We first show that the optimal terminal cushion ${C}_{T}^{\star}$ almost surely does not take values on $\left(-G,\hat{c}(G)\right)$ where the true and the concavified objective functions are different. Suppose that there is a set $\bar{\Omega}\subseteq\Omega$ such that for all $\omega\in\bar{\Omega}$,  $\tilde{U}\left({C}^{\star}_T(\omega);G\right)\neq\tilde{U}^{con}\left(C^{\star}_T\left(\omega\right);G\right)$, equivalently, ${C}^{\star}_T(\omega)$ takes values in $(-G,\hat{c}(G))$ for all $\omega\in\bar{\Omega}$, and ${C}^{\star}_T(\omega)>\hat{c}(G)$ for all $\omega \in \Omega \setminus \bar{\Omega}$. 
    We let $\mathbb{P}\left({C}_T^{\star}<\hat{c}(G)\right)=q$ for some $q\in(0,1)$. We want to contradict that such ${C}_T^{\star}$ is optimal. For all $\omega\in\bar{\Omega}$, there exists $\alpha(\omega)\in(0,1)$ such that ${C}_T^{\star}(\omega)=\alpha(\omega)\cdot\left(-G\right)+\left(1-\alpha\left(\omega\right)\right)\cdot\hat{c}(G)$. Now, define the random variable 
    \begin{equation*}
    \hat{C}_T\left(\omega\right)=
    \begin{cases}
    \begin{aligned}
    -G,\quad & \mbox{ if }\omega\in\Omega_1,\\
    \hat{c}(G),\quad & \mbox{ if }\omega\in\bar\Omega\setminus\Omega_1,\\
    {C}^{\star}_T\left(\omega\right),\quad & \mbox{ if }\omega\in\Omega\setminus \bar\Omega,    
    \end{aligned}
    \end{cases}
    \end{equation*}
    for some $\Omega_1\subset \bar \Omega$. We let $\mathbb{P}\left(\hat{C}_T=-G\right)=p_1$ for some $p_1\in[0,q]$ and $\mathbb{P}\left(\hat{C}_T=\hat{c}(G)\right)=q-p_1$. Note that $\tilde{U}\left(\hat{C}_T;G\right)=\tilde U^{con}\left(\hat{C}_T;G\right)$ everywhere.  Moreover 
    \begin{equation*}
    \mathbb{E}\left[\tilde U(\hat{C}_T;G)\right]= \tilde U(-G;G)p_1+\tilde{U}(\hat{c}(G);G)(q-p_1)+\mathbb{E}\left[\tilde{U}(\hat{C}_T;G)\mathds{1}_{\Omega\setminus\bar{\Omega}}\right].
    \end{equation*}
    On the other hand we also observe that 
    \begin{equation}
    \mathbb{E}\left[\tilde{U}({C}_T^{\star};G)\right]=\mathbb{E}\left[\tilde{U}\left(-G\alpha+\hat{c}(G)(1-\alpha);G\right)\mathds{1}_{\bar{\Omega}}\right]+\mathbb{E}\left[\tilde{U}(\hat{C}_T;G)\mathds{1}_{\Omega\setminus\bar\Omega}\right].
    \end{equation}
    Hence if we choose $p_1=0$ this implies that $\mathbb{E}\left[\tilde{U}(\hat{C}_T;G)\right]\ge\mathbb{E}\left[\tilde{U}({C}_T^{\star};G)\right]$. 
    Next we also exclude that ${C}^{\star}_T=-G$. Indeed, if one implements the trivial strategy, i.e. that corresponding to the multiplier $m=0$, gets the terminal cushion ${C}_T^0={c}_0\cdot e^{rT}$, which is non-negative since ${c}_0\geq 0$ by construction. Comparing terminal cushions ${C}_T^{\star}=-G$ $\mathbb{P}-$a.s. and $ C_T^0$, one would get that $\mathbb{E}\left[\tilde{U}({C}_T^{\star})\right]<\mathbb{E}\left[\tilde U({C}_T^{0})\right]$, which contradicts the optimality of ${C}_T^{\star}$. 
    \section{Proof of some technical results of Section \ref{sec:martingale_approach}}\label{app:B}
    \subsection{Proof of Proposition \ref{propos:state_price_decomp_3}}\label{app:B_6}
    Since the process $H^{\bm{\theta}}$ is Markovian, for all $\theta\in\tilde{\Theta}$, we can define, for all $t \in [0,T]$, the process $D\left(t,H^{\bm{\theta}}_{t}\right):=\mathbb{E}\left[\left(H^{\bm{\theta}}_{T}\right)^{-\frac{1-\delta_1}{\delta_1}}|\mathcal{F}_{t}\right]$.
    This is a martingale under $\mathbb{P}$, therefore, it can be characterized as the solution of the following PIDE
    \begin{align}\label{eq:PIDE_3}
    \nonumber 0=&\frac{\partial D\left(t,h\right)}{\partial t}+\dfrac{\partial D(t,h)}{\partial h}h\left[-r+\int_E\left(1-\theta^J_t(y)\right)\nu(\de y)\right]+\frac{1}{2}\frac{\partial^2 D(t,h)}{\partial h^2}h^2\left(\theta^D_t\right)^2\\
    &+\int_E\left(D(t,h+y)-D(t,h)\right)\nu(\de y),
    \end{align}
    with boundary condition $D\left(T,h\right)=h^{-\frac{1-\delta_1}{\delta_1}}$.
    To solve the problem depicted in equation \eqref{eq:PIDE_3}, we make the ansatz
    $D\left(t,h\right)=h^{-\frac{1-\delta_1}{\delta_1}}\tilde{D}\left(t\right)$, for all $ t\in[0,T]$, with $\tilde{D}\left(t\right)=e^{d_1(t)}$ for all $t\in[0,T]$ and $d_1(T)=0$. Substituting this ansatz in equation \eqref{eq:PIDE_3}, we obtain the following ODE
    \begin{align*}
    0=\frac{\de d_1(t)}{\de t}+\dfrac{1-\delta_1}{\delta_1}r-\int_E\left[\dfrac{1-\delta_1}{\delta_1}\left(1-\theta^J_t(y)\right)-\left(\theta^J_t(y)\right)^{-\frac{1-\delta_1}{\delta_1}}+1\right]\nu(\de y)+\frac{1}{2}\dfrac{1-\delta_1}{\delta_1^2}\left(\theta^D_t\right)^2,
    \end{align*} 
    whose well behaved solution up to time $T$, is given in equation \eqref{eq:tilde_G_process_1}.
    A sufficient condition for $D$ to be a martingale is that $
    \mathbb{E}\left\lbrace\int_0^{T} \frac{\left(1-\delta_1\right)^2}{\delta_1^2}\left(\theta^D_s\right)^2 \de s+ \int_0^{T}\int_E\bigg|\left(\theta^J_s(y)\right)^{-\frac{1-\delta_1}{\delta_1}}-1\bigg|\nu(\de y) \de s\right\rbrace<+\infty$.
    Note that under the Novikov condition, the condition on  $\theta^D$ is satisfied. Hence the right integrability of $\theta^J$ provides an additional sufficient condition for having $\bm{\theta} \in \tilde \Theta$.
    \subsection{Proof of Proposition \ref{propos:martingale_rapp_theo_1}}\label{app:B_2}
    Thanks to the decomposition of the state price density $H^{\bm{\theta}}$ given in Proposition \ref{propos:state_price_decomp_3}, we can rewrite $\mathcal{X}_{\bm{\theta}}(y)$ as
    $\mathcal{X}_{\bm{\theta}}(y)=y^{-\frac{1}{\delta_1}}\mathbb{E}\left[\left(H^{\bm{\theta}}_{T}\right)^{-\frac{1-\delta_1}{\delta_1}}\right]=y^{-\frac{1}{\delta_1}}\cdot D\left(0,H^{\bm{\theta}}_{0}\right)$.
    We denote by $\hat{y}_{\bm{\theta}}$ the unique point such that $\mathcal{X}_{\bm{\theta}}(\hat{y}_{\bm{\theta}})={c}_0$ for given ${c}_{0}$. Then $\hat{y}_{\bm{\theta}}=\left(\frac{{c}_{0}}{D\left(0,H^{\bm{\theta}}_{0}\right)}\right)^{-\delta_1}$, and 
    \begin{equation}\label{eq:rv_Y_1}
    Y_{\bm{\theta}}=\left({c}_{0}^{-\delta_1}H^{\bm{\theta}}_{T}\right)^{-\frac{1}{\delta_1}}=\frac{{c}_{0}}{D\left(0,H^{\bm{\theta}}_{0}\right)}\cdot\left(H^{\bm{\theta}}_{T}\right)^{-\frac{1}{\delta_1}}.
    \end{equation}
    Substituting equation \eqref{eq:rv_Y_1} into equation \eqref{eq:process_m}, we obtain
    \begin{equation}\label{eq:exp_M_1}
    M^{\bm{\theta}}_{t}=\frac{{c}_{0}}{D\left(0,H^{\bm{\theta}}_{0}\right)}\cdot\mathbb{E}\left[\left(H^{\bm{\theta}}_{T}\right)^{-\frac{1-\delta_1}{\delta_1}}\bigg|\mathcal{F}_{t}\right]=\frac{{c}_{0}D\left(t,H^{\bm{\theta}}_{t}\right)}{D\left(0,H^{\bm{\theta}}_{0}\right)}.
    \end{equation}
    
    By applying It\^{o}'s formula on equation \eqref{eq:exp_M_1}, we obtain the dynamics of $ M^{\bm{\theta}}_{t}$ which reads as follows
    \begin{equation}\label{eq:M_dynamics_1}
    \frac{\de M^{\bm{\theta}}_{t}}{M^{\bm{\theta}}_{t-}}=-\dfrac{1-\delta_1}{\delta_1}\theta^{D}_t\de W_t+\int_E\left[\left(\theta^J_t(y)\right)^{-\frac{1-\delta_1}{\delta_1}}-1\right]\left(N(\de t,\de y)-\nu(\de y)\de t\right).
    \end{equation}
    By comparing equation \eqref{eq:M_dynamics_1} with equation \eqref{eq:M_mrt_1} we obtain the result. 
    \subsection{Proof of Theorem \ref{teo:thm_1}}\label{app:B_3}
    To prove that ${C}_{t}={C}^{\hat{\bm{\theta}}}_{t}$ $\mathbb{P}-$a.s, we have to show that both processes evolves according to the same stochastic differential equation. Let $\hat{\bm{\theta}}\in\tilde{\Theta}$ and $\hat{m}^{\hat{\bm{\theta}}}\in\mathcal{M}$ satisfy the system of equations \eqref{eq:system_sol_1}. It follows from equation \eqref{eq:cushion_theta} that the dynamics of ${C}^{\hat{\bm{\theta}}}_{t}$ satisfies
    %\begin{equation}\label{eq:M_on_H_1}
    $\displaystyle \de{C}^{\hat{\bm{\theta}}}_{t}=\de\left(\frac{M^{\hat{\bm{\theta}}}_{t}}{H^{\hat{\bm{\theta}}}_{t}}\right)$.
    %\end{equation}
    Applying It\^{o}'s formula for jump-diffusion processes, and using equation \eqref{eq:M_mrt_1}, we have that 
    \begin{multline}
    \de{C}^{\hat{\bm{\theta}}}_{t}={C}^{\hat{\bm{\theta}}}_{t}r\de t+{C}^{\hat{\bm{\theta}}}_{t}\theta^D_t\frac{\theta^D_t}{\delta_1}\de t-{C}^{\hat{\bm{\theta}}}_{t}\frac{\theta^D_t}{\delta_1}\de W_t\\+{C}^{\hat{\bm{\theta}}}_{t-}\int_E\left[\left(\theta^J_t(y)\right)^{-\frac{1}{\delta_1}}-1\right]\left(N\left(\de t,\de y\right)-\theta^J_t\left(y\right)\nu(\de y)\de t\right). \qquad{} \label{eq:C_hat_1}
    \end{multline}
    Substituting equation \eqref{eq:system_sol_1} into equation \eqref{eq:C_hat_1} and, recalling that $\hat{\bm{\theta}}\in\tilde{\Theta}$ satisfies equation \eqref{eq:No_Arb_Cond1}, we obtain
    \begin{align*}
    \de{C}^{\hat{\bm{\theta}}}_{t}=&{C}^{\hat{\bm{\theta}}}_{t}\left[r+\hat m^{\hat{\bm{\theta}}}_t\left(\mu-r)\right)\right]\de t+{C}^{\hat{\bm{\theta}}}_{t}\hat m^{\hat{\bm{\theta}}}_t\sigma\de W_t+{C}^{\hat{\bm{\theta}}}_{t-}\hat m^{\hat{\bm{\theta}}}_{t-}\int_E\gamma(t,y)N\left(\de t,\de y\right),
    \end{align*}
    which coincides with the SDE in equation \eqref{eq:cushion_process} satisfied by the cushion process that follows the strategy $\hat m^{\hat{\bm{\theta}}}$. Moreover, since ${C}_0={C}^{\hat{\bm{\theta}}}_{0}={c}_{0}$, we conclude that ${C}_{t}={C}^{\hat{\bm{\theta}}}_{t}={c}_{0}$, $\mathbb{P}$-a.s., for all $t\in[0,T]$. Finally, we need to show that $\hat m^{\hat{\bm{\theta}}}$ is a solution to \eqref{eq:value_2}. Thanks to Lemma \ref{lem:Lemma_1} and equation \eqref{eq:inf_problem}, we only need to show that
    %\begin{equation}
    $\mathbb{E}\left[\frac{{C}_{T}^{1-\delta_1}}{1-\delta_1}\right]=\mathbb{E}\left[\frac{Y_{\hat{\bm{\theta}}}^{1-\delta_1}}{1-\delta_1}\right]$,
    %\end{equation}
    for $m=\hat m^{\hat{\bm{\theta}}}$, which is true if ${C}_{t}=Y_{\hat{\bm{\theta}}}$ $\mathbb{P}$-a.s. The latter equality comes from the facts that ${C}_{t}={C}^{\hat{\bm{\theta}}}_{t}$, for all $t\in[0,T]$, $\mathbb{P}$-a.s., and that by construction of ${C}^{\hat{\bm{\theta}}}$.
    \subsection{Proof of Proposition \ref{prop:existence_uniqueness}}\label{app:B_7}
    Let $\left(\tilde{\phi}_{1},\tilde{\phi}_{2}\right)$ be the range for the multiplier corresponding to each of the three cases (i)--(iii). Given  the optimality condition, we define the function
    $g\left(m\right):=\mu-r-\delta_1\sigma^2 m+\int_{\mathbb{R}\setminus\left\lbrace0\right\rbrace}\frac{\gamma(y)}{\left(1+\gamma(y)m\right)^{\delta_1}}\nu\left(\de y\right)$. Since $g'\left(m\right)<0$, then there exists a unique point $\hat{m}^{\hat{\bm{\theta}}}\in\left(\tilde{\phi}_{1},\tilde{\phi}_{2}\right)$ such that $g\left(\hat{m}^{\hat{\bm{\theta}}}\right)=0$ only if $g\left(\tilde{\phi}_{2}\right)\geq 0$, and $g\left(\tilde{\phi}_{1}\right)\leq 0$.  However, the latter relations are exactly the ones stated into the proposition,  and this concludes the proof.  
    \section{Proof of some technical results of Section \ref{sect:numeric}}\label{app:C}
    \subsection{Proof of Corollary \ref{cor:constant_jump_size}}\label{app:C_1}
    To prove the result we note that $\phi_1=\tilde{\gamma}\in(-1,0)$, and $\phi_2=0$, so the optimal multiplier $\hat m^{\hat{\bm{\theta}}}$ must take value in $\left(-\infty,-1/\tilde{\gamma}\right)$. We easily see that, $g'(m)=-\delta_1\sigma^2-\delta_1\tilde{\gamma}^2\left(1+\tilde{\gamma}m\right)^{-\delta_1}<0$,  and that
    \begin{equation*}
    \displaystyle\lim_{m\to-\infty}\mu-r-\delta_1\sigma^2m+\lambda\tilde{\gamma}\left(1+\tilde{\gamma}m\right)^{-\delta_1}=+\infty,\qquad\displaystyle\lim_{m\to-\tilde{\gamma}^{-1}}\mu-r-\delta_1\sigma^2m+\lambda\tilde{\gamma}\left(1+\tilde{\gamma}m\right)^{-\delta_1}=-\infty,
    \end{equation*}
    which implies that $\hat m^{\hat{\bm{\theta}}}$ exists and that it is the unique solution of equation \eqref{eq:optimal_mult_constant_jump_size}.
    \subsection{Proof of Corollary \ref{cor:KOU_MODEL}}\label{app:C_2}
    Since for Kou's model the jump sizes are such that $\phi_1=-1$ and $\phi_2=+\infty$, the optimal multiplier $\hat m^{\hat{\bm{\theta}}}$ if exists, lies in $[0,1]$. Moreover, since also in this case $g'(m)<0$ for all $m\in[0,1]$, the condition for existence and uniqueness stated in Proposition \ref{prop:existence_uniqueness} become $g(0)>0$, and $g(1)<0$. Hence, we have to compute $g(0)$ and $g(1)$, and find the conditions on Kou's model parameters such that the latter inequalities are satisfied. The expressions for $g(0)$ and $g(1)$ are given by
    \begin{align*}
    g(0)&=\mu-r+\lambda\int_{\mathbb{R}\setminus\left\lbrace 0\right\rbrace}\left(e^{y}-1\right)\left(p\eta_{+}e^{-\eta_{+}y}\mathds{1}_{\left\lbrace y\geq 0\right\rbrace}+\left(1-p\right)\eta_{-}e^{\eta_{-}y}\mathds{1}_{\left\lbrace y<0\right\rbrace}\right)\de y\\
    &=\mu-r-\frac{p\lambda}{1-\eta_{+}}-\frac{\left(1-p\right)\lambda}{1+\eta_{-}},
    \end{align*}
    \begin{align*}
    g(1)&=\mu-r-\delta_1\sigma^2+\lambda\int_{\mathbb{R}\setminus\left\lbrace 0\right\rbrace}\left(e^{y}-1\right)e^{-\delta_1y}\left(p\eta_1e^{-\eta_1y}\mathds{1}_{\left\lbrace y\geq 0\right\rbrace}+\left(1-p\right)\eta_2e^{\eta_2y}\mathds{1}_{\left\lbrace y<0\right\rbrace}\right)\de y\\
    &=\mu-r-\delta_1\sigma^2-\frac{\lambda p\eta_{+}}{\left(1-\delta_1-\eta_{+}\right)\left(\delta_1+\eta_{+}\right)}+\frac{\lambda(1-p)\eta_{-}}{\left(1-\delta_1+\eta_{-}\right)\left(\delta_1-\eta_{-}\right)},
    \end{align*}
    respectively. Thus to guarantee that $g(0)>0$, and $g(1)<1$, the conditions stated in equation \eqref{eq:g_0_Kou} and equation \eqref{eq:d_1_Kou} need to be satisfied. This ensures the existence and the uniqueness for an optimal multiplier $\hat m^{\hat{\bm{\theta}}}\in[0,1]$, and concludes the proof.
%    \bibliographystyle{plainnat}
%    \bibliography{Biblio}

    \end{document}